\documentclass[10pt, twocolumn]{IEEEtran}

\usepackage{graphicx,epsfig}
\usepackage[noadjust]{cite}
\usepackage{mcite}
\usepackage{amsfonts,helvet}
\usepackage{fancyhdr}
\usepackage{threeparttable}
\usepackage{epsf,epsfig}
\usepackage{amsthm}
\usepackage{amsmath}
\usepackage{siunitx}
\usepackage{amssymb}
\usepackage{dsfont}
\usepackage{subfigure}
\usepackage{color}
\usepackage{algorithm}
\usepackage{algpseudocode}
\usepackage{algcompatible}
\usepackage{enumerate}
\usepackage{cancel}

\newtheorem{theorem}{Theorem}

\newtheorem{corollary}{Corollary}

\newtheorem{lemma}{Lemma}

\newtheorem{remark}{Remark}

\begin{document}

\title{On the Optimal Feedback Rate in Interference-Limited Multi-Antenna Cellular Systems}

\author{Jeonghun~Park, \emph{Student Member, IEEE}, Namyoon~Lee, \emph{Member, IEEE}, Jeffrey G. Andrews, \emph{Fellow, IEEE}, and~Robert W. Heath Jr., \emph{Felow, IEEE}
\thanks{J. Park, J. G. Andrews, and R. W. Heath Jr. are with the Wireless Networking and Communication Group (WNCG), Department of Electrical and Computer Engineering, 
The University of Texas at Austin, TX 78701, USA. (E-mail: $\left\{\right.$jeonghun, rheath$\left\}\right.$@utexas.edu, jandrews@ece.utexas.edu)

N. Lee is with Dept. of Electrical Engineering, POSTECH, 77 Cheongam-Ro. Nam-Gu. Pohang. Gyeongbuk. Korea 37673 (Email:nylee@postech.ac.kr).

This research is supported in part by a gift from Huawei Technologies Co. Ltd. and the National Science Foundation under Grant No. NSF-CCF-1514275.}}

\maketitle \setcounter{page}{1} 

\begin{abstract}
We consider a downlink cellular network where multi-antenna base stations (BSs) transmit data to single-antenna users by using one of two linear precoding methods with limited feedback: (i) maximum ratio transmission (MRT) for serving a single user or (ii) zero forcing (ZF) for serving multiple users.
The BS and user locations are drawn from a Poisson point process, allowing
expressions for the signal-to-interference coverage probability and the ergodic spectral efficiency to be derived as a function of system parameters such as the number of BS antennas and feedback bits, and the pathloss exponent. 
We find a tight lower bound on the optimum number of feedback bits to maximize the \emph{net spectral efficiency}, which captures the overall system gain by considering both of  downlink and uplink spectral efficiency using limited feedback.
Our main finding is that,
when using MRT, the optimum number of feedback bits scales linearly with the number of antennas, and logarithmically with the channel coherence time. When using ZF, the feedback scales in the same ways as MRT, but also linearly with the pathloss exponent.
The derived results provide system-level insights into the preferred channel codebook size by averaging the effects of short-term fading and long-term pathloss.
\end{abstract}

\section{Introduction}
In multi-antenna cellular systems, particularly assuming frequency division duplex (FDD), two fundamental obstacles limit the gains in spectral efficiency:
\begin{itemize}
\item \emph{Inter-cell interference (ICI)}: In a cellular network that uses universal frequency reuse, ICI is unavoidable. Severe ICI leads to low signal-to-interference plus noise ratio (SINR) of the downlink users, an operating regime where the spatial multiplexing gain vanishes \cite{4407232}.
\item \emph{Limited feedback}: The finite rate of the feedback link restricts the precision of channel quantization, which leads to inaccurate channel state information at transmitter (CSIT). In the multi-user scenario, inaccurate CSIT causes the inevitable inter-user interference (IUI).
\end{itemize}
In spite of extensive research over a few decades, system level performance considering both ICI and limited CSI feedback has been challenging to characterize in a common framework. The major difficulty has been the deficiency of a tractable model that suitably captures the both effects of the ICI and limited CSI feedback. 
In this paper, we adopt two analytical tools to resolve the challenge: (i) network modeling by stochastic geometry \cite{andrew:10} for calculating the amount of ICI, and (ii) quantization cell approximation \cite{4014384, 4299617, 1512149, 1237136} for analyzing the channel quantization error. Using these tools, we characterize the coverage probability and the ergodic spectral efficiency of a multi-antenna downlink system with limited feedback, considering full ICI in the network.
Based on the characterization, we attempt to reveal the complicated interplay between the downlink performance, limited CSI feedback, and the ICI.

 \subsection{Related Work}
There has been extensive prior work on the downlink transmission rate as a function of a finite CSI feedback rate. In toy setups where a base station (BS) serves a single user in a single cell by ignoring ICI, downlink transmission rates were characterized as a function of the codebook size and the SNR. For example, in \cite{1237152, 1237136}, a channel codebook design method for a single-user transmission was proposed by using Grassmannian line packing. 
By employing random vector quantization (RVQ), the rate loss of a point-to-point MIMO system caused by using finite rate feedback was characterized in \cite{1624652}. In \cite{4100151}, by applying RVQ, the performance of conjugate beamforming was analyzed in a MISO system.
In \cite{4770163}, the performance of the finite rate CSI feedback was characterized, assuming temporally correlated channel and a principle of designing a channel codebook for this particular condition was proposed. 
The main limitation of \cite{1237152,1237136, 4770163,1624652,4100151} is that it assumed only a single pair of a BS and a user, in which important features of a cellular network, e.g., ICI are missing.

Considering a single-cell and multi-user transmission scenario, it was shown in \cite{1715541} that the sum spectral efficiencies can increase without bound with the transmit power, provided that the CSI feedback rate linearly scales with SNR in decibels. In \cite{4299617}, by leveraging multi-user diversity gain, it was shown that a semi-orthogonal user selection method \cite{1603708} with finite rate feedback for zero-forcing (ZF) beamforming achieves the same throughput gain with the case in perfect CSIT. In \cite{5466522}, the achievable spectral efficiency was characterized by considering not only the finite rate feedback, but also the channel estimation error via downlink channel training.
In \cite{5671564}, multi-mode MIMO transmission depending on the users' channel conditions was proposed. 
In another line of research, in \cite{1577000, 5773636, 6112149}, the cost of the CSI feedback was considered. Specifically, the downlink spectral efficiency obtained by using limited CSI feedback was normalized by the uplink spectral efficiency spent by sending the feedback. Similar to
\cite{1237152,1237136, 1468321, 4770163,1624652,4100151},
the major limitation of the aforementioned work \cite{1715541,4299617,1603708,5466522,5671564,1577000,5773636, 6112149} is the use of over-simplified network model that only captures the effect of the channel quantization error, ignoring ICI; thereby, the results do not necessarily apply to cellular systems where ICI is significant. 


The effect of ICI on limited feedback has been addressed in prior work assuming deterministic BS locations. In \cite{akoumHeath10}, the the impact of ICI and delay on single user limited feedback was characterized. In \cite{5755206, 5613944, 5648782}, feedback bit allocation methods were proposed to balance resolution between intra-cluster interference channels for multi-user multi-cell coordinated beamforming. The results presented in \cite{akoumHeath10, 5755206, 5613944, 5648782} only holds for specific user locations, making it difficult to provide a system level analysis over many user locations and system parameters. 
 
An approach to model ICI is to leverage stochastic geometry, where spatial locations of BSs and users were modeled by using a homogeneous Poisson point process (PPP). Analytical expressions of the downlink cellular SIR performance were characterized by averaging the BSs' and the users' locations \cite{andrew:10}. 
Using the random network model based on PPP to model an ad hoc network, the transmission capacity was obtained when a multi-antenna transmission technique is used with limited CSI feedback in \cite{6205588, 5754751}.
In \cite{6410048, junzhang:dynamic}, BS cooperation methods with limited feedback were proposed and the SIR performance was analyzed.
There is prior work \cite{dhillon:ordering, gupta:hetnet, renzo:mimo, renzo:eid, chang:mimohetnet} that investigated the performance of MIMO in a network model built upon stochastic geometry, especially in a heterogeneous network \cite{dhillon:ordering, gupta:hetnet, chang:mimohetnet}.  
This work, however, did not consider limited feedback which is the main subject of this paper. 
In a network modeled by PPP, how the feedback rate scales with which system parameters is still a question in a fundamental cellular system where one BS serves one or multiple users through a linear precoding.
The relevant results under the previous assumptions do not hold since the SNR of the individual user is averaged by tools of stochastic geometry. 
This paper proposes an answer for such a question.

\subsection{Contributions}
In this paper, we characterize the downlink performance of a multi-antenna cellular system with limited CSI feedback. 
We consider two cases of interest, (i) single-user maximum ratio transmission (MRT) and (ii) multi-user ZF where the number of users is equal to the number of BS antennas. First we establish exact expressions for the SIR coverage probability and the ergodic spectral efficiency in integral forms as a function of the relevant system parameters: the pathloss exponent, the number of BS antennas, and the number of feedback bits. 
Subsequently, we obtain a lower bound on the optimum number of feedback bits that maximizes the \emph{net spectral efficiency}, which measures the normalized downlink gain for one channel coherence block.
The specific definition of the net spectral efficiency is provided in \eqref{net_per}.
Our key findings are summarized as follows.

Assume that $N$ is the number of BS antennas, $B$ is the number of feedback bits,  $\beta$ is the pathloss exponent, and $T_{\rm c}$ is the channel coherence time, specifically defined as the number of downlink symbols that experience the same channel fading. In single-user MRT with limited CSI feedback, for large enough $T_{\rm c}$, the optimum number of feedback bits is approximately
\begin{align} \label{intro_mrt}
B^{\star}_{\rm MRT}\approx (N-1) \log_2\left(T_{\rm c} \right),
\end{align} 
whereas in multi-user ZF with limited CSI feedback, 
for large enough $T_{\rm c}$, 
\begin{align} \label{intro_zf}
B^{\star}_{\rm ZF} \approx (N-1)\frac{\beta}{2}\log_2\left(T_{\rm c} \right).
\end{align}
In both \eqref{intro_mrt} and \eqref{intro_zf}, the optimum number of feedback bits scales linearly with the number of antennas and logarithmically with the channel coherence time, while it also scales linearly with the pathloss exponent for multi-user ZF. 
Neither expression is a function of the instantaneous SIR because all the randomness affecting SIR, e.g., short-term fading and long-term pathloss effects are averaged into the derived forms.

The paper is organized as follows. Section II introduces the models and the performance metrics. In what follows, we characterize the performance and find a lower bound on the optimum number of feedback bits for single-user MRT in Section III and for multi-user ZF in Section IV, respectively.
Section V shows simulation results for verifying the obtained results and Section VI concludes the paper.

\section{Models and Metrics}
\subsection{Network Model}
A downlink cellular network model is considered, where BSs equipped with $N$ antennas are distributed according to a homogeneous PPP, $\Phi = \left\{{\bf{d}}_i,i\in \mathbb{N} \right\}$ with density $\lambda$. In this network model, single antenna users are distributed as an independent homogeneous PPP, $\Phi_{\rm{U}} = \left\{ {\bf{u}}_i,i\in \mathbb{N} \right\}$ with density $\lambda_{{\rm{U}}}$. 
Each BS in the network has a coverage region characterized by a Voronoi tessellation. 



\subsection{Signal Model}
In each cell, $K$ users are selected to be served. For $\lambda_{{\rm U}} \gg \lambda$, at least $K$ users are in each Voronoi region with high probability, therefore each BS is able to choose $K$ users out of users located in each cell. 
We analyze performance for the user located at the origin, denoted as user $k$, $1 \le k \le K$, per Slivnyak's theorem \cite{baccelli:inria}. User $k$ is served by the BS located at ${\bf d}_1$, which is the nearest BS to the origin in $\Phi$. Other users $k'$, $k' \in \{1,...,K\}\backslash k$ are also served by the BS located at ${\bf{d}}_1$.
The BS located at ${\bf{d}}_1 $ sends $K$ information symbols to its respective $K$ associated users through a precoding matrix ${\bf{V}}_1\in\mathbb{C}^{N\times K}$,
so the received signal at user $k$ is given by
\begin{align} \label{sig_model}
y_k = \left\| {\bf{d}}_1 \right\|^{-\beta/2} {\bf{h}}_{k,1}^* {\bf{V}}_1 {\bf{s}}_1 + \sum_{i=2}^{\infty} \left\| {\bf{d}}_i \right\|^{-\beta/2} {\bf{h}}_{k,i}^* {\bf{V}}_i {\bf{s}}_i + z_k,
\end{align}
where ${\bf{h}}_{k,i} \in \mathbb{C}^{N \times 1}$ is the channel coefficient vector from the BS at ${\bf{d}}_i$ to user $k$,
${\bf{V}}_i = \left[{\bf{v}}_1^i,..., {\bf{v}}_K^i\right] \in \mathbb{C}^{N \times K}$ and $\left\| {\bf{v}}^i_k \right\| = 1$ for $k\in \{1,...,K\}$ is a precoding matrix of the BS at ${\bf{d}}_i$, ${\bf{s}}_i =\left[s_1^i,...,s_K^i \right]^T \in \mathbb{C}^{K \times 1}$ is an information symbol vector transmitted from the BS at ${\bf{d}}_i$, and $z_k \sim \mathcal{CN}\left(0,\sigma^2\right)$ is additive Gaussian noise. The pathloss exponent is $\beta > 2$.
We assume that $\mathbb{E}\left[{\bf{s}}_i {\bf{s}}_i^* \right] = P/N \cdot {\bf{I}}$ for $i \in \mathbb{N}$ and ${\bf{I}}$ is the identity matrix.
Each entry of the channel coefficient vector ${\bf{h}}_{k,i}$
is drawn from  independent and identically distributed (IID) complex Gaussian random variables, i.e., $\mathcal{CN}\left(0,1\right)$. Depending on the number of users, we consider two cases of interest: single-user MRT when $K=1$ and multi-user ZF when $K=N$. 
While the case of $1 < K < N$ is also of interest, it is less tractable and left to future work.
Henceforth, we drop the BS index $1$ for simplification. For instance, we will write the channel coefficient vector ${\bf{h}}_{k,1}$ as ${\bf{h}}_k$, the precoding matrix ${\bf{V}}_1 = \left[{\bf{v}}_1^1,..., {\bf{v}}_K^1\right] $ as ${\bf{V}} = \left[{\bf{v}}_1,..., {\bf{v}}_K\right] $. 

\subsection{Feedback Model}
Before data transmission, user $k$ learns downlink CSI ${\bf{h}}_k$ by using predefined pilot symbols sent from the BSs, and sends it back to the associated BS via a finite rate feedback link. To do this, a user quantizes the channel direction information (CDI) by using a predefined codebook that is known to both of the associated BS and user $k$. Assuming that feedback link capacity is $B$ bits, the cardinality of the codebook $\mathcal{C}$ is $2^B$, where each entry of the codebook is selected from $N$-dimensional unit norm vectors, i.e., $\mathcal{C} = \left\{{\bf{w}}_1,...,{\bf{w}}_{2^B} \right\}$ and $\left\| {\bf{w}}_i \right\| = 1$ for $i \in \left\{1,...,2^B \right\}$. By measuring the inner products between the channel direction vector $\tilde {\bf{h}}_k = {\bf{h}}_k / \left\| {\bf{h}}_k \right\|$ and the codeword vectors ${\bf w}_i$ for $i \in \{1,...,2^B\}$, a user chooses an index that provides the maximum inner product value, 
\begin{align}
i_{\max} = \mathop {\arg \max} \limits_{i = 1,...,2^B} \left| \tilde {\bf{h}}_k ^* {\bf{w}}_i \right|.
\end{align}
The chosen index $i_{\max}$ is sent to the BS. Since the BS has the same codebook as user $k$, it acquires quantized channel direction information $\hat {\bf{h}}_k = {\bf{w}}_{i_{\max}}$ from $i_{\max} $. The quantized channel information $\hat {\bf{h}}_k$ is used for designing a precoding matrix ${\bf{V}}$.

Although it is also possible to quantize the channel quality information (CQI), i.e., $\left\| {\bf{h}}_k \right\|$ and send it to the BS through the feedback link, we only consider the CDI feedback. 
If user scheduling is considered, then CQI feedback becomes more important and the results might become different.
For instance, the CQI and the CDI can be jointly exploited to select a better set of users by employing the semi-orthogonal user selection algorithm \cite{1603708}.
Additionally, in \cite{6601769}, it was also shown that the sum-rate performance is better when large portions of the feedback bits are used for CQI. We leave this issue as future work.

For analytical tractability in characterizing performance of the limited feedback strategy as a function of the codebook size, we use the quantization cell approximation technique  \cite{4014384, 4299617, 1512149, 1237136}, which assumes that each quantization cell is a Voronoi region of a spherical cap. 
This is a standard approach in vector quantization \cite{1056067, gersho:book} and it is used to deal with the irregular shape of the Voronoi quantization regions. 
In this paper, we call this technique the spherical-cap approximation of vector quantization (SCVQ).
Using SCVQ, the area of a quantization cell is $2^{-B}$ when the number of feedbacks is $B$. This assumption leads to the following approximation of the cumulative distribution function (CDF) of quantization error \cite{4299617}.
\begin{align} \label{sin_cdf}
F_{{\rm sin^2}\theta_k}\left(x\right) = \left\{\begin{array}{cc}2^Bx^{N-1}, &  0\le x \le \delta \\ 1, & \delta \le x \end{array}, \right.
\end{align}
where ${\rm sin^2}\theta_k = 1 - \left|\tilde {\bf{h}}_k^* \hat{\bf{h}}_k \right|^2$ and $\delta = 2^{-\frac{B}{N-1}}$. In \cite{4299617}, for any quantization codebook that has a quantization error CDF $F_{{\rm sin^2}\tilde \theta_k}\left(x\right)$, we have $F_{{\rm sin^2}\theta_k}\left(x\right) \ge F_{{\rm sin^2}\tilde \theta_k}\left(x\right)$. 
Due to this property, SCVQ provides an upper bound performance with limited CSI feedback.
Nevertheless, the SCVQ is useful in analyzing the quantization error effects accurately, as the downlink rate performance with the SCVQ is known to be tight with that using RVQ, which provides a lower bound performance of the limited feedback strategy \cite{4299617}.  

\subsection{Performance Metrics}
In this section, we define the SIR, the SIR complementary cumulative distribution function (CCDF), and the ergodic spectral efficiency for the two cases we consider in this paper, i.e., the single-user MRT and the multi-user ZF. We focus on the SIR instead of the SINR, since cellular systems are usually interference limited \cite{andrew:10}.
We also assume that all the cells in a network employ the same beamforming strategy. For example, in single-user MRT, each cell serves a single user by using MRT beamforming. 
\subsubsection{Single-User MRT}
In single-user MRT, only one user is active, i.e., $K=1$. In this case, the BS designs a precoding matrix to maximize desired channel gain.
Since perfect channel information ${\bf{h}}_1$ is impossible to obtain, the actual channel gain $\left|{{{\bf{h}}}_1^*} {\bf{v}}_1 \right|^2$ cannot be used to design ${\bf{v}}_1$. Instead, by exploiting the quantized channel direction information $\hat {\bf{h}}_1$ obtained from CSI feedback, the quantized channel gain $\left|{\hat{{\bf{h}}}_1^*} {\bf{v}}_1 \right|^2$ is used.
To maximize this, ${\bf{v}}_1$ is designed as ${\bf{v}}_1 = \hat {\bf{h}}_1$. Applying this, the signal model \eqref{sig_model} is rewritten as 
\begin{align}
y_1 = \left\| {\bf{d}}_1 \right\|^{-\beta/2} {\bf{h}}_{1}^* \hat{{\bf{h}}}_1 {{s}}_1 + \sum_{i=2}^{\infty} \left\| {\bf{d}}_i \right\|^{-\beta/2} {\bf{h}}_{1,i}^* {{\bf{v}}}_i {{s}}_i + z_k,
\end{align}
The instantaneous SIR of the typical user is given by
\begin{align} \label{sir_su}
{\rm{SIR}}_{{\rm{MRT}}} = \frac{P \left\| {\bf{d}}_1 \right\|^{-\beta} \left|{\bf{h}}_1^*  \hat{{\bf{h}}}_1 \right|^2 }{P \sum_{i=2}^{\infty}\left\| {\bf{d}}_i \right\|^{-\beta} \left| {\bf{h}}_{1,i}^*{{\bf{v}}}_i \right|^{2} }.
\end{align}
With \eqref{sir_su}, we define the CCDF of the instantaneous SIR with the target SIR $\gamma$ as follows.
\begin{align} \label{sir_ccdf_su}
P_{\rm MRT} = \mathbb{P}\left[{\rm SIR_{MRT}} > \gamma \right].
\end{align}
The ergodic spectral efficiency of single-user MRT is defined as
\begin{align} \label{dfn_rate_su}
R_{\rm{MRT}} = \mathbb{E}\left[\log_2\left(1 + {\rm{SIR}}_{{\rm{MRT}}} \right) \right].
\end{align}
The expectation is taken over the multiple randomness including fadings, locations of BSs and users, and the realizations of a random codebook. 

\subsubsection{Multi-User ZF}
In multi-user ZF, $K=N \ge 2$ users are selected for receiving information symbols. For multi-user communication, the ZF beamformer is used. 
Since we only consider the case where the number of active users is the same as the number of BS antennas, the only objective of the ZF beamformer is mitigating the IUI, but it is not used for increasing the desired signal power. The more general case $K<N$ is also of interest, but we leave it as future work.
With the quantized CSI ${\hat{\bf{h}}}_k$, $k\in \{1,...,K\}$, the precoding vector ${\bf{v}}_k$ is designed to satisfy
\begin{align}
\hat {\bf{h}}_{k'}^*{\bf{v}}_k = 0,\;\;k' \in \{1,...,K\} \backslash k.
\end{align}
Since the obtained CSI is not perfect, i.e., ${\hat{\bf{h}}}_k \ne \tilde{{\bf{h}}}_k$, the IUI is not perfectly removed. Considering this, the signal model \eqref{sig_model} is rewritten as
\begin{align}
y_k &= \left\| {\bf{d}}_1 \right\|^{-\beta/2} {\bf{h}}_{k}^* {\bf{v}}_k {{s}}_k + 
\sum_{k' = 1, k'\ne k}^{K}\left\| {\bf{d}}_1 \right\|^{-\beta}  {\bf{h}}_{k}^* {\bf{v}}_{k'} s_{k'} \nonumber \\
&+\sum_{i=2}^{\infty} \left\| {\bf{d}}_i \right\|^{-\beta/2} {\bf{h}}_{k,i}^* {\bf{V}}_i {\bf{s}}_i + z_k.
\end{align}
Denoting the remaining inter-user interference and the inter-cell interference as $I_{\rm U} = P/K \sum_{ k' =1,k' \ne k}^{K}  \left\| {\bf{d}}_1 \right\|^{-\beta} \left| {\bf{h}}_{k}^* {\bf{v}}_{k'} \right|^2$ and $I_{\rm C} = P/K\sum_{i=2}^{\infty}\left\| {\bf{d}}_i \right\|^{-\beta} \left\| {\bf{h}}_{k,i}^*{\bf{V}}_i \right\|^{2}$, 
the instantaneous SIR of the typical user $k$ is given by
\begin{align} \label{sir_mu}
{\rm{SIR}}_{{\rm{ZF}}}^k = \frac{P/K \left\| {\bf{d}}_1 \right\|^{-\beta} \left| {\bf{h}}_k^*  {\bf{v}}_k \right|^2 }{I_{\rm U} + I_{\rm C} }.
\end{align}

By leveraging the instantaneous SIR, the CCDF of the instantaneous SIR is defined as
\begin{align} \label{sir_ccdf_mu}
P_{\rm ZF}^{k} = \mathbb{P}\left[{\rm SIR}_{\rm ZF}^{k} > \gamma \right].
\end{align}
The ergodic spectral efficiency is also defined as
\begin{align} \label{dfn_rate_mu}
R_{\rm ZF}^k = \mathbb{E} \left[\log_2\left(1 + {\rm{SIR}}_{\rm ZF}^{k} \right) \right].
\end{align}

\subsection{Net Spectral Efficiency}
We define the net spectral efficiency to measure the difference between downlink and uplink spectral efficiencies. This is essential for evaluating the overall system gain, as the downlink spectral efficiency improvement comes at the cost of the uplink spectral efficiency spent by sending CSI feedback.
When a user sends $B$ feedback bits, the net spectral efficiency is defined as
\begin{align} \label{net_per}
 R_{\rm Net}\left(B\right) =R\left(B\right) - B/T_{\rm c},
\end{align}
where $R\left(B\right)$ is the downlink ergodic spectral efficiency and $T_{\rm c}$ is the channel coherence time, specifically defined as the number of downlink symbols that experience the same channel fading. 
From \eqref{net_per}, the net spectral efficiency is penalized by $B/T_{\rm c}$ when using $B$ feedback bits, so that $1/T_{\rm c}$ behaves as a penalization factor.
The reason for choosing $1/T_{\rm c}$ as the penalization factor is as follows.
Rewriting \eqref{net_per} by multiplying $T_{\rm c}$ on each side, we have
\begin{align} \label{eq:net_tc}
T_{\rm c} R_{\rm Net}(B) = T_{\rm c} R(B) - B.
\end{align}
In \eqref{eq:net_tc}, we find that the net spectral efficiency multiplied with $T_{\rm c}$ consists of the sum downlink spectral efficiency corresponding to one channel coherence block, subtracted by the number of feedback bits used for quantizing the corresponding channel. From this observation, the net spectral efficiency \eqref{net_per} measures the normalized net gain for one channel coherence block when using $B$ feedback bits. 
In this case, when $T_{\rm c}$ is large, the number of feedback bits should be increased since the inaccuracy of the CSIT can cause a significant effect to a large number of downlink symbols. 
We also note that this interpretation was also adopted in \cite{1577000, 5773636}. 
In this paper, we assume $T_{\rm c} > 100$, which means that more than $100$ downlink symbols experience the same channel fading.

\section{Single-User Maximum Ratio Transmission}
In this section, we analyze the CCDF of the instantaneous SIR and the ergodic spectral efficiency for MRT. Based on the derived expressions, we obtain a lower bound on the optimum number of feedback bits that maximizes the net spectral efficiency $R_{\rm Net}(B)$.

\subsection{SIR CCDF Characterization}
We first attempt to derive the CCDF of the instantaneous SIR as an integral form. 
Theorem \ref{theo_sir_ccdf_su} is presented for the main result of this subsection.

\begin{theorem} \label{theo_sir_ccdf_su}
When the number of feedback bits is $B$, the CCDF of the instantaneous SIR of single-user MRT, defined in \eqref{sir_ccdf_su}, is
\begin{align} \label{theo_sir_ccdf_su_main}
\mathbb{P}\left[{\rm SIR_{MRT}} > \gamma \right] =  \sum_{m=0}^{N-1} \frac{\gamma^m}{m!} (-1)^m \left. \frac{\partial^m \mathcal{L}_{I/\cos^2\theta_1}(s)}{\partial s^m} \right| _{s = {\gamma}},
\end{align}
where $I = \left\| {\bf{d}}_1 \right\|^{\beta}   { \sum_{i=2}^{\infty}\left\| {\bf{d}}_i \right\|^{-\beta} \left| {\bf{h}}_{1,i}^*{{{\bf{V}}}}_i \right|^{2} }$, ${\rm cos}^2 \theta_1 = \left|{\tilde{\bf{h}}}_1 \hat{\bf{h}}_1 \right|^2$ and $\mathcal{L}_{I/\cos^2\theta_1}(s)$ is the Laplace transform of the random variable $I/\cos^2\theta_1$, which is given by
\begin{align} \label{laplace_I_cos}
&\mathcal{L}_{I/\cos^2\theta_1}(s) = \nonumber \\
&\int_{0}^{2^{-\frac{B}{N-1}}} \frac{\beta - 2}{\beta - 2 + 2\frac{s}{1 - x} \cdot {}_2F_1\left(1, \frac{-2+\beta}{\beta}, 2- \frac{2}{\beta}, -\frac{s}{1-x} \right)} \cdot \nonumber \\
&\;\;\;\;\;\;\;\;\;\;\;\;\;\;  2^B (N-1)x^{N-2} \ {\rm d} x.
\end{align}
$_2F_1\left(\cdot, \cdot, \cdot, \cdot \right) $ is the Gauss-hypergeometric function defined as
\begin{align}
_2F_1\left(a,b,c,z \right) = \frac{\Gamma(c)}{\Gamma(b)\Gamma(c-b)}\int_{0}^{1} \frac{t^{b-1}(1-t)^{c-b-1}}{(1-tz)^{a}} {\rm d} t.
\end{align}
\end{theorem}
\begin{proof}
See Appendix \ref{proof:thm1}.
\end{proof}

The obtained CCDF of the instantaneous SIR expression contains the relevant system parameters, e.g., the number of BS antennas $N$, the pathloss exponent $\beta$, and the number of feedback bits $B$. 

\subsection{Ergodic Spectral Efficiency Characterization}

Before deriving the ergodic spectral efficiency, we derive Lemma \ref{lem_s_mgf_su} which  provides the Laplace transform of the desired channel gain in terms of the number of feedback bits $B$, and refer Lemma \ref{lem_useful} which presents an integral form of the ergodic spectral efficiency.

\begin{lemma} \label{lem_s_mgf_su}
In single-user MRT, the Laplace transform of the desired channel gain $\left| {\bf{h}}_1^* \hat{{\bf{h}}}_1 \right|^2$ is
\begin{align}
\mathbb{E}\left[e^{-s\left| {\bf{h}}_1^* \hat{{\bf{h}}}_1 \right|^2} \right] &= \left(\frac{1}{1+s}\right) \left(\frac{1}{1+s\left(1 - 2^{-\frac{B}{N-1}} \right)} \right)^{N-1}.
\end{align}
\end{lemma}
\begin{proof}
From the definition of the Laplace transform, 
\begin{align}
&\mathbb{E}\left[e^{-s\left| {\bf{h}}_1^* {\hat{\bf{h}}}_1 \right|^2} \right] 
=\mathbb{E}\left[e^{-s\left\| {\bf{h}}_1 \right\|^2 \left| \tilde {\bf{h}}_1^* \hat {\bf{h}}_1 \right|^2} \right] 
= \mathbb{E}\left[e^{-s\left\| {\bf{h}}_1 \right\|^2 \cos^2 \theta_1}  \right] \nonumber \\
&\mathop {=} \limits^{(a)} \mathbb{E}_{\cos\theta_1}\left[\left(\frac{1}{1+s\cos^2\theta_1} \right)^N \right] \nonumber \\
&= \mathbb{E}_{\sin\theta_1}\left[\left(\frac{1}{1+s\left( 1 - \sin^2 \theta_1\right)} \right)^N \right] \nonumber \\
&\mathop = \limits^{(b)} \left(\frac{1}{1+s}\right) \left(\frac{1}{1+s\left(1 - 2^{-\frac{B}{N-1}} \right)} \right)^{N-1},
\end{align}
where (a) follows from $\left\|{\bf{h}}_1 \right\|^2 \sim \chi^2_{2N}$, and (b) follows \eqref{sin_cdf}.
\end{proof}

Notice that when random beamforming is used, which requires no CSI feedback, the desired channel gain $H_1^{{\rm RB}}=\left| {\bf{h}}_1^* {\bf{v}}_1 \right|^2$ is distributed as Chi-squared with degrees of freedom two. The corresponding Laplace transform, therefore, is $\mathbb{E}\left[e^{-sH_1^{{\rm RB}}} \right]=\frac{1}{1+s}$. Whereas, when perfect CSIT is used for MRT, the desired channel power denoted as  $H_1^{{\rm CSIT}}=\left| {\bf{h}}_1^* {\bf{v}}_1 \right|^2$ is distributed as Chi-squared with $2N$ degrees of freedom; so its Laplace transform is $\mathbb{E}\left[e^{-sH_1^{{\rm CSIT}}} \right]=\left(\frac{1}{1+s}\right)^N$. As a result, we confirm that the Laplace transform of the desired channel power with limited CSI feedback is lower and upper bounded by the two Laplace transforms, namely, 
\begin{align}
\frac{1}{1+s} \geq \left(\frac{1}{1+s}\right) \left(\frac{1}{1+s\left(1 - 2^{-\frac{B}{N-1}} \right)} \right)^{N-1} \geq\left(\frac{1}{1+s}\right)^N
\end{align}
for all $s\geq 0$, $B\geq 1$, and some $N\geq 1$. This relationship clearly exhibits that the MRT gain with limited feedback is larger than random beamforming gain from a Laplace transform perspective. For example, the additional gain of MRT with limited CSI feedback compared to random beamforming is $\left({1}/{1+s\left(1 - 2^{-\frac{B}{N-1}} \right)} \right)^{N-1}$ by using $B$ feedback bits. Further, we show that the Laplace transform of MRT with limited feedback converges to that of MRT with perfect CSIT as $B$ goes infinity.


Next, Lemma \ref{lem_useful} is derived to give an integral form of the ergodic spectral efficiency. 
\begin{lemma}[\cite{useful}, Lemma 1] \label{lem_useful}
Let $x_1,...,x_N,y_1,...,y_M$ be arbitrary non-negative random variables. Then
\begin{align}
&\mathbb{E}\left[\ln \left( 1+ \frac{\sum_{n=1}^{N}x_n  }{\sum_{m=1}^{M} y_m + 1} \right) \right]  \nonumber \\
&= \int_{0}^{\infty} \frac{\mathcal{M}_y\left(z\right) - \mathcal{M}_{x,y}\left(z\right)}{z}{\rm{exp}}\left(-z\right) {\rm d}z,
\end{align}
where $\mathcal{M}_y\left(z\right) = \mathbb{E}\left[e^{-z\sum_{m=1}^{M}y_m} \right]$ and $\mathcal{M}_{x,y}\left(z\right) = \mathbb{E} \left[ e^{-z\left( \sum_{n=1}^{N} x_n + \sum_{m=1}^{M} y_m \right)} \right]$.
\end{lemma}
\begin{proof}
See Lemma 1 in the reference \cite{useful}.
\end{proof}

Now, we obtain the ergodic spectral efficiency for single-user MRT in an integral form.
\begin{theorem} \label{rate_sing}
The ergodic spectral efficiency of the single-user MRT with $B$ bits feedback is  
\begin{align} \label{erate_single}
&\mathbb{E}\left[\log_2 \left(1 + {\rm{SIR}}_{\rm MRT} \right) \right]  \nonumber\\
&= \log_2 e\int_{0}^{\infty}\frac{1}{z}\left(1 - \left(\frac{1}{1+z}\right) \left(\frac{1}{1+z\left(1 - 2^{-\frac{B}{N-1}} \right)} \right)^{N-1} \right) \nonumber \\
&\;\;\;\;\;\;\;\;\;\;\;\;\;\;\;\;\; \cdot \left( \frac{1}{1 + 2z\cdot \frac{ {}_2F_1\left(1, \frac{-2+\beta}{\beta}, 2- \frac{2}{\beta}, -z \right)}{\beta-2}}\right) {\rm d} z.
\end{align}
\end{theorem}
\begin{proof}
From Lemma \ref{lem_useful}, we have
\begin{align} \label{thm1_proof1}
&\mathbb{E}\left[\log_2 \left(1 + {\rm{SIR}}_{\rm MRT} \right) \right] \nonumber \\
&= \log_2 e \int_{0}^{\infty} \frac{1}{z} \left(1 - \mathcal{M}_{ S}\left(z\right) \right)\mathcal{M}_{ I}\left(z\right)  {\rm d} z,
\end{align}
where $\mathcal{M}_{ S}\left(z\right) = \mathbb{E}\left[e^{-z\left|{\bf{h}}_1^* {\bf{v}}_1 \right|^2} \right]$ and $\mathcal{M}_{ I}\left(z\right) = \mathbb{E}\left[e^{-z\left\| {\bf{d}}_1 \right\|^{\beta}\sum_{i=2}^{\infty}\left\| {\bf{d}}_i \right\|^{-\beta} \left| {\bf{h}}_{1,i}^*{\bf{V}}_i \right|^{2}} \right]$. 
Since $\mathcal{M}_{ S}\left(z\right)$ is obtained in Lemma \ref{lem_s_mgf_su}, we calculate $\mathcal{M}_{ I}\left(z\right)$. By following the same step as \eqref{laplace_su_cos}, we have
\begin{align}
&\mathbb{E}\left[e^{-z\left\| {\bf{d}}_1 \right\|^{\beta}\sum_{i=2}^{\infty}\left\| {\bf{d}}_i \right\|^{-\beta} \left| {\bf{h}}_{1,i}^*{\bf{V}}_i \right|^{2}} \right] \nonumber \\
&\mathop {=} \limits^{} \frac{\beta - 2}{\beta - 2 + 2z \cdot {}_2F_1\left(1, \frac{-2+\beta}{\beta}, 2- \frac{2}{\beta}, -z \right)}.
\end{align}
Plugging $\mathcal{M}_{ S}\left(z\right)$ and $\mathcal{M}_{ I}\left(z\right)$ into \eqref{thm1_proof1}, the proof is completed.
\end{proof}
When the number of feedback bits $B$, the number of antenna $N$, and the pathloss exponent $\beta$ are given, we are able to obtain the ergodic spectral efficiency for single-user MRT by calculating \eqref{erate_single} numerically. 
The verification of Theorem \ref{rate_sing} will be provided in Section V.

While Theorem \ref{rate_sing} provides the exact ergodic spectral efficiency, it is not easy to see the benefit of increasing $B$ to the ergodic spectral efficiency since the characterization is in an integral form. To resolve this, we provide a lower bound on the ergodic spectral efficiency in the following corollary.

\begin{corollary} \label{lower_sing_coro}
The ergodic spectral efficiency of single-user MRT is lower bounded by
\begin{align}
&\mathbb{E}\left[\log_2\left( 1+{\rm SIR}_{\rm MRT}\right) \right] \nonumber \\
&\ge \log_2\left(1+\left(1-2^{-\frac{B}{N-1}}\right)\frac{\exp\left(\psi(N)  \right)}{2/(\beta-2)} \right),
\end{align}
where $\psi(\cdot)$ is the digamma function defined as
\begin{align} \label{def_digamma}
\psi(x) = \int_{0}^{\infty}\left(\frac{e^{-t}}{t} - \frac{e^{-xt}}{1-e^{-t}}\right) {\rm d} t.
\end{align}
\end{corollary}
\begin{proof}
See Appendix \ref{proof:coro1}. 
\end{proof}

Corollary \ref{lower_sing_coro} reveals the effect of the feedback bits in a way that is easier to understand compared with Theorem \ref{rate_sing}. Intuitively, increasing the number of feedback bits $B$ improves the SIR term by $\left(1-2^{-\frac{B}{N-1}} \right)$ in the spectral efficiency \eqref{dfn_rate_su}.

\subsection{Lower Bound on the Optimum Number of Feedback Bits}
Now we derive a lower bound on the number of feedback bits $B^{\star}_{\rm MRT}$ that maximizes the net spectral efficiency defined in \eqref{net_per}. 

\begin{theorem} \label{theo_low_sing}
In single-user MRT, the number of feedback bits $B^{\star}_{\rm MRT}$ that maximizes the net spectral efficiency is lower bounded by
\begin{align} \label{lower_sing_claim}
B^{\star}_{\rm MRT} & \ge B^{\star}_{\rm L, MRT} \nonumber \\
&=  \left( N-1\right) \log_2\left(\frac{\left(\beta - 2\right) N + (\beta -2)T_{\rm c}}{\left(\beta - 2 \right) N + \beta} \right). 
\end{align}
\end{theorem}
\begin{proof} 
See Appendix \ref{proof:thm3}.
\end{proof}

The verification of the obtained analytical lower bound will be provided in Section V.

\begin{remark} \label{remark_qub}
\normalfont
In this paper, we use the SCVQ technique, which provides an upper bound performance over all channel codebooks \cite{1512149}. Conversely, the feedback bits obtained under the SCVQ assumption indicates a general lower bound on cases of applying other practical channel codebooks. 
The reason is that, when using a practical codebook, more feedback bits are required to have the same performance with the SCVQ assumption as the SCVQ assumption gives an upper bound performance. For this reason, the obtained lower bound $B_{\rm L, MRT}^{\star}$ serves as an actual lower bound for any quantization codebook.
\end{remark}

One important observation in Theorem \ref{theo_low_sing} is that the lower bound $B^{\star}_{\rm L, MRT}$ is not a function of $\lambda$. This agrees with the result of \cite{andrew:10}, which revealed that the SIR of the cellular network modeled by a homogenous PPP is independent to $\lambda$. Since the ergodic spectral efficiency is a function of the SIR, it is also independent to $\lambda$, which leads to the result that $B^{\star}_{\rm L,MRT}$ is independent to $\lambda$.

Next, we find an intuitive approximation of $B^{\star}_{\rm L,MRT}$ when $T_{\rm c} \gg 1$ in the following corollary.
\begin{corollary} \label{coro_sing_approx}
When $N, \beta$ are given and $T_{\rm c}$ is large enough, a lower bound on the optimum number of feedback bits, denoted as $B^{\star}_{\rm MRT}$, is approximately
\begin{align} 
B^{\star}_{\rm MRT} \ge B^{\star}_{\rm L, MRT} \approx (N-1)\log_2\left(T_{\rm c}\right).
\end{align}
\end{corollary}
\begin{proof}
From Theorem \ref{theo_low_sing}, a lower bound $B_{\rm L}^{\star}$ is
\begin{align}
B^{\star}_{\rm MRT} \ge B^{\star}_{\rm L,MRT} =  \left( N-1\right) \log_2\left(\frac{\left(\beta - 2\right) N + (\beta -2)T_{\rm c}}{\left(\beta - 2 \right) N + \beta} \right).
\end{align}
When $T_{\rm c} \gg 1$, we have the following approximation.
\begin{align}
&\left( N-1\right) \log_2\left(\frac{\left(\beta - 2\right) N + (\beta -2)T_{\rm c}}{\left(\beta - 2 \right) N + \beta} \right) \nonumber \\
& \mathop {\approx} \limits^{(a)} \left( N-1\right) \log_2\left(\frac{(\beta -2)T_{\rm c}}{\left(\beta - 2 \right) N + \beta} \right) \nonumber \\
& = \left( N-1\right) \log_2\left(T_{\rm c} \right) + (N-1)\log_2\left(\frac{(\beta - 2)}{\left(\beta - 2\right) N + \beta} \right) \nonumber \\
&\mathop {\approx} \limits^{(b)} \left( N-1\right) \log_2\left(T_{\rm c} \right),
\end{align}
where (a) comes from $T_{\rm c} \gg (\beta-2)N$, and (b) comes from $\log_2\left(T_{\rm c}\right) \gg \log_2\left(\frac{(\beta-2)}{\left(\beta-2 \right)N+\beta} \right)$.
This completes the proof.
\end{proof}
From Corollary \ref{coro_sing_approx}, we observe that the optimum number of feedback bits scales linearly with the number of antennas and  scales logarithmically with the channel coherence time. 
As mentioned before, $B^{\star}_{\rm MRT}$ is not a function of instantaneous SIR, so it provides an appropriate channel codebook size which is not changed depending on short-term channel fading or even long-term pathloss.

\section{Multi-User Zero Forcing}

In this section, we characterize the CCDF of the instantaneous SIR and the ergodic spectral efficiency for multi-user ZF with the feedback bits. After that, as in previous section, we derive a lower bound on the optimum number of feedback bits $B^{\star}_{\rm ZF}$ that maximizes the net spectral efficiency.

\subsection{SIR CCDF Characterization}

Following the same steps with single-user MRT, we derive the CCDF of the instantaneous SIR for ZF. For better understanding, we first present ${\rm{SIR}}_{{\rm{ZF}}}^k$ \eqref{sir_mu} as
\begin{align} \label{sir_mu2} 
&{\rm{SIR}}_{{\rm{ZF}}}^k \nonumber \\
&=  \frac{P/K \left\| {\bf{d}}_1 \right\|^{-\beta} \left| {\bf{h}}_k^*  {\bf{v}}_k \right|^2 }{I_{\rm U}   + I_{\rm C} }
\nonumber \\
&= \frac{\frac{P}{K} \left| {\bf{h}}_k^*  {\bf{v}}_k \right|^2 }{\frac{P}{K} \sum_{ k' =1,k' \ne k}^{K}   \left| {\bf{h}}_{k}^* {\bf{v}}_{k'} \right|^2   +   \frac{P}{K}\left\| {\bf{d}}_1 \right\|^{\beta}  \sum_{i=2}^{\infty}\left\| {\bf{d}}_i \right\|^{-\beta} \left\| {\bf{h}}_{k,i}^*{\bf{V}}_i \right\|^{2} } \nonumber \\
& \mathop = \limits^{(a)} \frac{\left\| {\bf{h}}_k \right\|^2 \beta\left(1, N-1\right) }{\Bigg(\begin{array}{c} \left\| {\bf{h}}_k \right\|^2 {\rm sin}^2 \theta_k  \sum_{ k' =1,k' \ne k}^{K} \beta\left(1, N-2 \right)   \\
+   \left\| {\bf{d}}_1 \right\|^{\beta}  \sum_{i=2}^{\infty}\left\| {\bf{d}}_i \right\|^{-\beta} \left\| {\bf{h}}_{k,i}^*{\bf{V}}_i \right\|^{2} \end{array} \Bigg) }, 
\end{align}
where (a) follows (21) in \cite{4299617}. In the above, $\beta\left(1, M \right)$ is a Beta random variable that follows ${\rm Beta} \left(1, M \right)$. 
Now, Lemma \ref{lem_beta} is presented for the distribution of a product of a Gamma random variable and a Beta random variable.

\begin{lemma} \label{lem_beta}
Let $G$ and $B$ be random variables that follow the Gamma distribution $\Gamma\left( M, \theta\right)$ and the Beta distribution ${\rm{Beta}}\left(1, M-1\right)$, respectively. Then,  $H = GB$ is an exponential random variable with mean $\lambda = 1/\theta$.
\end{lemma}
\begin{proof}
See Theorem 1 in \cite{prodgammabeta}.
\end{proof}

Leveraging Lemma \ref{lem_beta}, Theorem \ref{theo_sir_ccdf_mu} gives the CCDF of the instantaneous SIR of multi-user ZF. One should note that we assume independence between the desired signal and the IUI terms for analytical tractability. This assumption was also used in \cite{6205588} and shown to be reasonable.
The verification provided in the next section also shows that the results under this assumption matches well with the simulation results.

\begin{theorem} \label{theo_sir_ccdf_mu}
The CCDF of the instantaneous SIR for multi-user ZF with $B$ bits feedback is
\begin{align}
&\mathbb{P}\left[{\rm SIR}_{\rm ZF}^k > \gamma \right] \nonumber \\
&= \left( \frac{1}{1+\gamma  2^{-\frac{B}{N-1}}} \right)^{N-1} \frac{1}{{}_2F_1\left(N, -\frac{2}{\beta}, 1-\frac{2}{\beta}, -\gamma \right)}.
\end{align}
\end{theorem}
\begin{proof}
See Appendix \ref{proof:thm4}.
\end{proof}
In Theorem \ref{theo_sir_ccdf_mu}, it is clear that the CCDF of the instantaneous SIR of multi-user ZF consists of two separate terms, representing the Laplace transforms of IUI and ICI, respectively. 

\subsection{Ergodic Spectral Efficiency Characterization}

\begin{theorem} \label{mu_rate}
The ergodic spectral efficiency of the multi-user ZF transmission with $B$ bits feedback is
\begin{align}
&\mathbb{E}\left[\log_2 \left(1 + {\rm{SIR}}^k_{\rm MU} \right) \right]  \nonumber\\
&= \log_2 e\int_{0}^{\infty} \left( \frac{1}{1+z}\right) \left( \frac{1}{1+z 2^{-\frac{B}{N-1}}}  \right)^{N-1} \cdot \nonumber \\
&\;\;\;\;\;\;\;\;\;\;\;\;\;\;\;\;\;\;\;\;\;\;\;\;\;\;\;\; \left(\frac{1}{{}_2F_1\left(N, -\frac{2}{\beta}, 1-\frac{2}{\beta}, -z \right)} \right){\rm d} z.
\end{align}
\end{theorem}
\begin{proof}
The proof is similar to that of the single-user case. We start from Lemma \ref{lem_useful}.
\begin{align} 
&\mathbb{E} \left[\log_2\left(1 + {\rm{SIR}}_{\rm MU}^{k} \right) \right] \nonumber \\
&= \log_2 e \int_{0}^{\infty} \frac{1}{z} \left(1 - \mathcal{M}_{ S}\left(z\right) \right)\mathcal{M}_{ I}\left(z\right)  {\rm d} z,
\end{align}
where $\mathcal{M}_{ S}\left(z\right) = \mathbb{E}\left[e^{-z\left|{\bf{h}}_k^* {\bf{v}}_k \right|^2} \right]$ and 
\begin{align}
&\mathcal{M}_{ I}\left(z\right) \nonumber \\
&= \mathbb{E}\left[e^{-z\left(I_{\rm U} + I_{\rm C}\right) } \right]\nonumber\\
&= \mathbb{E}\left[e^{-z\left( \sum_{ k' =1,k' \ne k}^{K}   \left| {\bf{h}}_{k}^* {\bf{v}}_{k'} \right|^2   +   \left\| {\bf{d}}_1 \right\|^{\beta}  \sum_{i=2}^{\infty}\left\| {\bf{d}}_i \right\|^{-\beta} \left\| {\bf{h}}_{k,i}^*{\bf{V}}_i \right\|^{2} \right)} \right] \nonumber \\
&= \mathcal{L}_{I_{\rm U}}(z) \mathcal{L}_{I_{\rm C}}(z).
\end{align}
As revealed in Lemma \ref{lem_beta}, $\left| {\bf{h}}_k^* {\bf{v}}_k \right|^2  = \left\| {\bf{h}}_k \right\|^2 \beta\left(1, N-1\right)  \sim \Gamma \left(1,1 \right)$ so we have the Laplace transform $\mathcal{M}_{ S}\left(z\right) = 1/\left(1+z\right)$. Plugging the Laplace transform of the IUI and ICI obtained in \eqref{laplace_iui_sir_ccdf_final} and \eqref{laplace_ici_sir_ccdf_final} into $\mathcal{M}_{ I}\left(z\right)$, we complete the proof.
\end{proof}

The verification for Theorem \ref{mu_rate} will be provided in Section V.
We also provide a lower bound on the ergodic spectral efficiency of ZF in the following corollary.

\begin{corollary} \label{rate_mu_lower}
The ergodic spectral efficiency of the multi-user ZF with $B$ bits feedback is lower bounded by
\begin{align}
\mathbb{E}\left[\log_2\left( 1+{\rm SIR}_{\rm ZF}^k\right) \right] \ge \log_2\left(1+\frac{\exp(\psi(1))}{(N-1)2^{-\frac{B}{N-1}} + \frac{2N}{\beta-2}} \right).
\end{align}
\end{corollary}
\begin{proof}
By using the lower bound \eqref{ny_lowerbound}, we have
\begin{align} \label{coro_mu_lower}
&\mathbb{E}\left[\log_2\left(1+ {\rm SIR}_{\rm ZF}^{k} \right) \right] \nonumber \\
& \ge\log_2\Bigg(1+ \frac{ \overbrace{\exp\left( \mathbb{E}\left[ \ln \left(\left\| {\bf{h}}_k \right\|^2 \beta\left(1, N-1\right) \right) \right] \right)}^{(a)}}{\left(\begin{array}{c} \underbrace{ \mathbb{E}\left[ \left\| {\bf{h}}_k \right\|^2 {\rm sin}^2 \theta_k  \sum_{ k' =1,k' \ne k}^{K} \beta\left(1, N-2 \right)\right]}_{(b)}   \\
+  \underbrace{\mathbb{E} \left[ \left\| {\bf{d}}_1 \right\|^{\beta}  \sum_{i=2}^{\infty}\left\| {\bf{d}}_i \right\|^{-\beta} \left\| {\bf{h}}_{k,i}^*{\bf{V}}_i \right\|^{2}\right]}_{(c)} \end{array} \right)} \Bigg),
\end{align}
By leveraging the distribution obtained in Lemma \ref{lem_beta}, expectations in the log function are calculated as 
${\rm (a)} = \exp\left(\psi(1) \right)$,
${\rm (b)} = (N-1)2^{-\frac{B}{N-1}}$, and ${\rm (c)} = \frac{2N}{\beta-2}$.
Combining the calculation results completes the proof.
\end{proof}
In Corollary \ref{rate_mu_lower}, it is well observed that increasing the feedback bits $B$ boosts the ergodic spectral efficiency by mitigating the IUI in inverse power scale, i.e., $2^{-\frac{B}{N-1}}$.

\subsection{Approximate Lower Bound on the Optimum Number of Feedback Bits}
In this subsection, we derive an approximate lower bound on the optimum number of feedback bits that works for a particular region of $T_{\rm c}$, specifically $T_{\rm c} \gg 1$. Later, we numerically show that the gap between the derived approximation and the numerically obtained optimum feedback bits is tight in our interest range of $T_{\rm c}$.
Before the derivation, we first provide the following lemma for obtaining a lower bound on the Laplace transform of the ICI.

\begin{lemma} \label{lem_lower_MGF_ICI}
The Laplace transform of the ICI is lower bounded as follows.
\begin{align}
\mathcal{L}_{I_{\rm C}}(s) &= \frac{1}{{}_2F_1\left(N, -\frac{2}{\beta}, 1-\frac{2}{\beta}, -s \right)}   \nonumber \\
&\ge \frac{1}{1+s^{\frac{2}{\beta}}\frac{2}{\beta}N^{\frac{2}{\beta}}\left(-\Gamma(-\frac{2}{\beta}) \right)}.
\end{align}
\end{lemma}
\begin{proof}
Rewriting $\mathcal{L}_{I_{\rm C}}(s)$, we have
\begin{align}
&\frac{1}{{}_2F_1\left(N, -\frac{2}{\beta}, 1-\frac{2}{\beta}, -s \right)}  \nonumber \\
&=  \frac{1}{1+\left(2\int_{1}^{\infty} \left(1-\left(\frac{1}{1+st^{-\beta}} \right)^N \right)t{\rm d} t \right)} \nonumber \\
& \ge \frac{1}{1+\left(2\int_{0}^{\infty} \left(1-\left(\frac{1}{1+st^{-\beta}} \right)^N \right)t{\rm d} t \right)} \nonumber \\
& \mathop {\ge} \limits^{(a)} \frac{1}{1+\left(2\int_{0}^{\infty} \left(1-e^{-s Nt^{-\beta}} \right)t{\rm d} t \right)} \nonumber \\
& \mathop = \limits^{(b)} \frac{1}{1+\left(\frac{2}{\beta}s^{\frac{2}{\beta}} N^{\frac{2}{\beta}}\int_{0}^{\infty} \left(1-e^{-u} \right)u^{-1-\frac{2}{\beta}}{\rm d} u \right)} \nonumber \\
& = \frac{1}{1+s^{\frac{2}{\beta}}\frac{2}{\beta}N^{\frac{2}{\beta}}\left(-\Gamma(-\frac{2}{\beta}) \right)},
\end{align}
where (a) follows $1/(1+st^{-\beta}) \ge e^{-st^{-\beta}}$, thereby
\begin{align}
1-\left(\frac{1}{1+st^{-\beta}}\right)^N \le 1-e^{-sNt^{-\beta}},
\end{align}
and (b) comes from variable change $sNt^{-\beta} = u$. This completes the proof.
\end{proof}

By leveraging Lemma \ref{lem_lower_MGF_ICI}, we derive an approximate lower bound in the following theorem. 

\begin{theorem} \label{them_low_approx}
As $T_{\rm c} \gg 1$, a lower bound on the optimum feedback bits is approximated as
\begin{align}
B^{\star}_{\rm ZF} \ge B^{\star}_{\rm L,ZF} & \approx \tilde B^{\star}_{\rm L,ZF} \nonumber \\
&= \left(N-1 \right)\log_2\left(\frac{\beta \Gamma\left(1-\frac{2}{\beta}\right) T_{\rm c}}{ 2 N \left(-\Gamma(-\frac{2}{\beta}) \right)}\right)^{\frac{\beta}{2}}.
\end{align}
\end{theorem}
\begin{proof}
See Appendix \ref{proof:thm6}.
\end{proof}
The verification for Theorem \ref{them_low_approx} will be provided in Section V. 

Next, to find an intuitive form of $\tilde{B}^{\star}_{\rm L, ZF}$, we more approximate $\tilde {B}^{\star}_{\rm L, ZF}$ in the following corollary, 
\begin{corollary} \label{coro_mu_approx}
For $T_{\rm c} \gg 1$, an approximate lower bound on the optimum feedback bits, denoted as $\tilde{B}^{\star}_{\rm L, ZF}$, is further approximated as
\begin{align}
B^{\star}_{\rm ZF} \ge B^{\star}_{\rm L, ZF} \approx \tilde B^{\star}_{\rm L, ZF} \approx (N-1)\frac{\beta}{2}\log_2\left(T_{\rm c} \right).
\end{align}
\end{corollary}
\begin{proof}
The proof is similar to that of Corollary \ref{coro_sing_approx}. When $T_{\rm c} \gg 1$, we have the following approximation.
\begin{align}
&\tilde B^{\star}_{\rm L,ZF} \nonumber \\
&= \left(N-1 \right)\log_2\left(\frac{\beta \Gamma\left(1-\frac{2}{\beta}\right) T_{\rm c}}{ 2 N \left(-\Gamma(-\frac{2}{\beta}) \right)}\right)^{\frac{\beta}{2}} \nonumber \\
&= (N-1) \log_2\left( T_{\rm c}\right)^{\frac{\beta}{2}} + (N-1) \log_2\left( \frac{\beta \Gamma\left( 1-\frac{2}{\beta}\right)}{2N \left(-\Gamma\left(-\frac{2}{\beta} \right) \right)}\right)^{\frac{\beta}{2}} \nonumber \\
& \mathop {\approx}^{(a)} (N-1) \log_2\left( T_{\rm c}\right)^{\frac{\beta}{2}},
\end{align}
where (a) comes from that
\begin{align}
(N-1) \log_2\left( T_{\rm c}\right)^{\frac{\beta}{2}} \gg (N-1) \log_2\left( \frac{\beta \Gamma\left( 1-\frac{2}{\beta}\right)}{2N \left(-\Gamma\left(-\frac{2}{\beta} \right) \right)}\right)^{\frac{\beta}{2}} 
\end{align}
when $T_{\rm c}$ is large enough. This completes the proof.
\end{proof}

Similar to single-user MRT, $\tilde{B}^{\star}_{\rm L, ZF}$ is not a function of instantaneous 
SIR since it averages over all the randomness that affects the SIR. 
The different feature from single-user MRT is that the optimal number of feedback bits increases with the pathloss exponent linearly. A reasonable explanation for this observation can be found in the definition of the net spectral efficiency.  
Considering the net spectral efficiency, a major factor that determines the optimum number of feedback bits is the feedback efficiency, which measures how much downlink ergodic spectral efficiency improves when increasing the small number of feedback bits. When the pathloss exponent is small, the typical user has a large amount of the ICI, therefore the operating SIR is chiefly determined by the dominant ICI. In this case, the feedback efficiency is likely to be low since the downlink spectral efficiency only marginally improves even when mitigating the IUI by increasing the number of feedback bits. For this reason, the optimum number of feedback bits is small.
In the opposite case, if the pathloss exponent is large, the typical user has less amount of ICI. Then, the IUI mainly determines the operating SIR. When the IUI decreases by increasing the number of feedback bits, the downlink spectral efficiency would improve dramatically compared to the case of a small pathloss exponent. As a result, the optimum number of feedback bits is large.

On the contrary, in single-user MRT, the feedback information is used to boost up the desired signal power, so that the number of feedback bits $B$ is only related to
the desired signal term while 
not affecting the interference term.
For this reason, normalizing the SIR \eqref{sir_su} by the desired signal's pathloss
allows us to represent the SIR as the ratio between a function of $B$ and a function of $\beta$. 
We write it as 
\begin{align}
R_{\rm MRT} &= \log_2\left(1+\frac{S(B)}{I(\beta)} \right).
\end{align}
Then, for large $T_{\rm c}$, we approximate the above as 
\begin{align} \label{approx:mrt_beta}
\log_2\left(1+\frac{S(B)}{I(\beta)} \right)  &\mathop {\approx}^{(a)} \log_2\left(\frac{S(B)}{I(\beta)} \right) \nonumber \\
&= \log_2(S(B)) - \log_2(I(\beta)).
\end{align}
The high SIR assumption (a) is justified as follows: for large $T_{\rm c}$, using many feedback bits is encouraged, so the desired signal power is relatively high compared to the inter-cell interference. 
When differentiating \eqref{approx:mrt_beta} with respect to $B$, the $\beta$-related term vanishes, thereby the feedback efficiency is independent of $\beta$,
and the optimum feedback rate is also independent of $\beta$. 

\section{Numerical Comparisons}
In this section, we verify the obtained analytical expressions by comparing to  simulation results. We assume that the BSs are distributed as a homogeneous PPP with $\lambda = 10^{-5}/\pi$ and that the pathloss exponent $\beta = 4$. As a channel quantization method, we use RVQ, which serves a lower bound performance on the limited feedback strategy. 
We assume Rayleigh fading, where each channel coefficient is drawn from IID complex Gaussian random variables. The number of iterations is $5000$. 

First, we verify the analytical results of single-user MRT. Fig.~\ref{qub_rvq_mrt} shows the comparison between Theorem \ref{rate_sing} and the simulation results. 
As observed in this figure, the performance gap between SCVQ and RVQ is less than 1 $\rm bps/Hz$, and this gap becomes vanish as the number of used feedback bits increases. 
Fig.~\ref{sin_verify} illustrates the gap between an analytical lower bound $B^{\star}_{\rm L, MRT}$ and the exact $B^{\star}_{\rm MRT}$ obtained numerically. 
We use $\lceil B^{\star}_{\rm L, MRT} \rceil$ to draw the analytical lower bound.
In the graph, the gap between numerically obtained $B^{\star}_{\rm MRT}$ and $B^{\star}_{\rm L, MRT}$ is tight over the entire range of $T_{\rm c}$, therefore $B^{\star}_{\rm L, MRT}$ is a good indicator of $B^{\star}_{\rm MRT}$. 

\begin{figure}[!t]
\centerline{\resizebox{0.8\columnwidth}{!}{\includegraphics{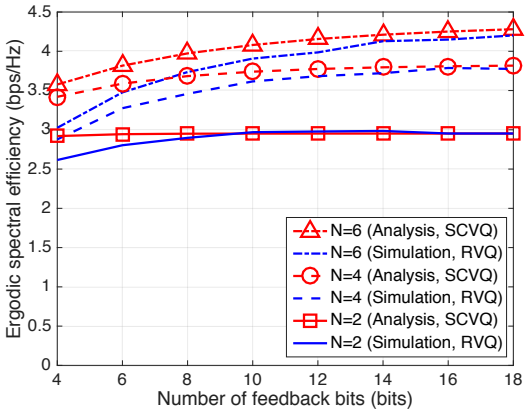}}}    
\caption{Single-user MRT ergodic spectral efficiency comparison vs. the number of feedback bits sharing both analysis with SCVQ and simulation with RVQ. It is assumed that $\lambda = 10^{-5}/\pi$ and $\beta = 4$.}
 \label{qub_rvq_mrt}
\end{figure}

\begin{figure}[!t]
\centerline{\resizebox{0.84\columnwidth}{!}{\includegraphics{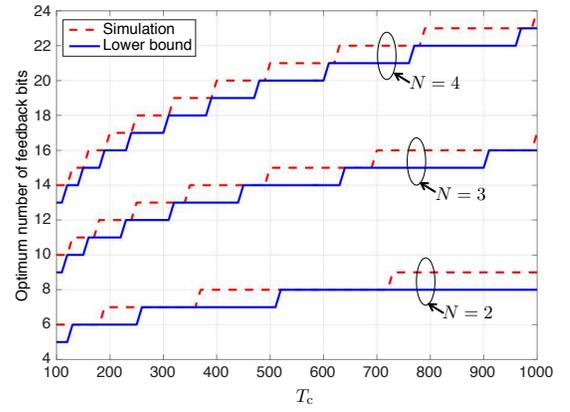}}}    
\caption{Illustration for verfifying an analytical lower bound $B^{\star}_{\rm L, MRT}$ by comparing the exact $B^{\star}_{\rm MRT}$ obtained numerically. It is assumed that $\beta = 4$.}
 \label{sin_verify}
\end{figure}

Next, we provide verification of the analytical results of multi-user ZF. 
As shown in Fig.~\ref{qub_rvq_zf}, the gap between the analysis and simulation is less than $1 {\rm bps/Hz}$, and the gap becomes smaller as the number of feedback bits increases. 
Fig.~\ref{mul_verify} shows the comparison between $ B^{\star}_{\rm ZF}$ obtained numerically and $\tilde B^{\star}_{\rm L,ZF}$ derived in Theorem \ref{them_low_approx}.
When drawing the figure, we use $\lceil \tilde B^{\star}_{\rm L,ZF}\rceil$.
As observed in the figure, the gap between the optimum feedback bits $B^{\star}_{\rm ZF}$ and the analytical approximation $\tilde B^{\star}_{\rm L, ZF}$ is small over all range of $T_{\rm c}$. 

\begin{figure}[!t]
\centerline{\resizebox{0.84\columnwidth}{!}{\includegraphics{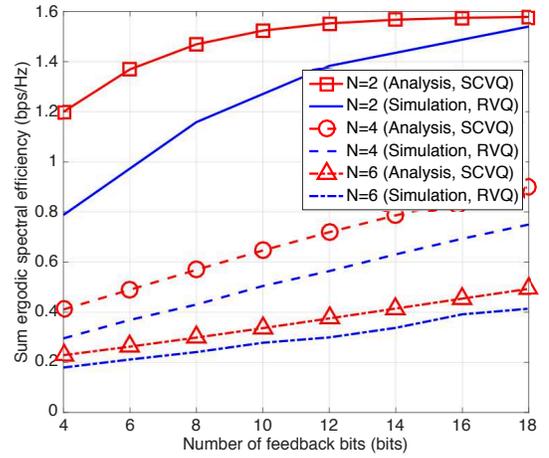}}}    
\caption{
Multi-user ZF ergodic spectral efficiency comparison vs. the number of feedback bits sharing both analysis with SCVQ and simulation with RVQ. It is assumed that $\lambda = 10^{-5}/\pi$ and $\beta = 4$.}
 \label{qub_rvq_zf}
\end{figure}

\begin{figure}[!t]
\centerline{\resizebox{0.84\columnwidth}{!}{\includegraphics{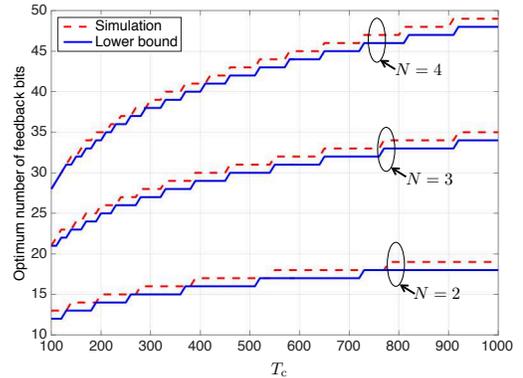}}}    
\caption{Illustration for verfifying an analytical lower bound $\tilde B^{\star}_{\rm L, ZF}$ by comparing the exact $B^{\star}_{\rm ZF}$ obtained numerically. It is assumed that $\beta = 4$.}
 \label{mul_verify}
\end{figure}

Finally, we numerically compare the ergodic spectral efficiency of single-user MRT and multi-user ZF when the number of feedback bits are provided as $B^{\star}_{\rm L, MRT}$ and $\tilde {B}^{\star}_{\rm L, ZF}$, respectively. Considering the net spectral efficiency, each transmission method achieves their maximum performance when using $B^{\star}_{\rm L, MRT}$ and $\tilde {B}^{\star}_{\rm L, ZF}$.
\begin{figure}[!t]
\centerline{\resizebox{0.85\columnwidth}{!}{\includegraphics{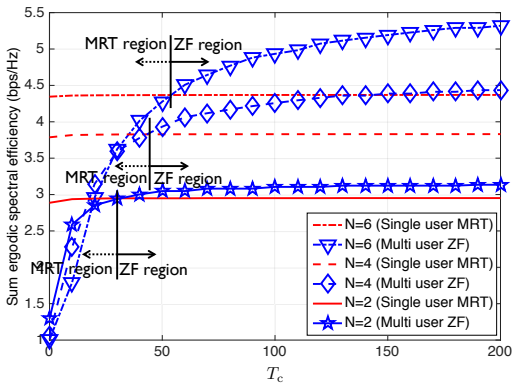}}}    
\caption{Sum ergodic spectral efficiency comparison between  single-user MRT and multi-user ZF. In each method, the feedback bits are provided as $B^{\star}_{\rm L, MRT}$ and $\tilde{B}^{\star}_{\rm L, ZF}$, respectively. It is assumed that $\beta = 4$.}
 \label{mrt_vs_zf}
\end{figure}
In Fig.~\ref{mrt_vs_zf}, it is observed that there exists a threshold of $T_{\rm c}$ that separates the whole region into two regions. In each region, called the MRT region and the ZF region, the maximum performance of one transmission method is better than that of the other transmission method. For instance, as $T_{\rm c}$ increases, the sum ergodic spectral efficiency of multi-user ZF dominates that of single-user MRT. From this observation, we can conclude that if $T_{\rm c}$ larger than a particular threshold, using multi-user ZF for extracting the multiplexing gain is more beneficial than single-user MRT for the array gain.

\section{Conclusions}

In this paper, we derived a tight lower bound on the optimum number of feedback bits $B^{\star}_{\rm L, MRT}$ for single-user MRT and $\tilde{B}^{\star}_{\rm L, ZF}$ for multi-user ZF when deployed in a cellular network. The main results are summarized as $B^{\star}_{\rm L, MRT} \approx (N-1)\log_2\left(T_{\rm c} \right)$ and $\tilde{B}^{\star}_{\rm L, ZF} \approx (N-1)\frac{\beta}{2}\log_2\left(T_{\rm c} \right)$. 
In both cases, the optimum number of bits scales linearly with the number of antennas and logarithmically with the length of the coherence block. The multi-user ZF approach requires typically more feedback (usually $\beta>2$) since higher resolution of CSIT is required to mitigate the limited feedback-induced interference.
These are new findings, compared to prior analyses of limited feedback. For example, in
\cite{1715541}, assuming multi-user ZF, it was shown that the number of feedback bits should scale as $(N-1)\log_2\left({\rm SNR}\right)$ to achieve a constant performance gap to the perfect CSIT case. This result holds under the assumption of a deterministic user location by treating the ICI as additive Gaussian noise. 
An implication of the result of \cite{1715541} is that the channel codebook size depends on the users' locations, which requires adaptive codebooks to implement unlike our result. 
In this sense, our results propose more general channel codebook size than \cite{1715541} by averaging the effects of not only short-term channel fading but also long-term pathloss determined by a user location. 
However, one should note that stochastic averaging over the SIR is not the reason that $T_{\rm c}$ appears in the obtained expressions of $B^{\star}_{\rm L, MRT}$ or $\tilde{B}^{\star}_{\rm L, ZF}$. The main reason for this is our particular selection of the uplink penalization factor as $1/T_{\rm c}$. 
If a different penalty function is considered, then the results may change.

It is also worthwhile to compare with other work that used limited feedback and considered ICI. In \cite{5755206, 5648782, 6410048}, feedback bit allocation methods were proposed for multi-user multi-cell coordinated beamforming under the assumption of deterministic BSs locations. The core idea is to allocate the feedback bits to cooperative BSs in proportion to their received power. The rationale is that the limited feedback achieves high efficiency by assigning more feedback bits to 
BSs whose signal power is relatively strong.
For this reason, when the pathloss exponent increases, the fraction of feedback dedicated to the primary serving (the closest) BS increases since the relative signal power coming from the far distance diminishes quickly.
Our result assumes random users and BSs locations, but we observe a similar relationship on the pathloss exponent in Corollary \ref{coro_mu_approx}. When the pathloss exponent increases, the ICI power decays drastically, and the IUI power becomes relatively strong. This causes more feedback bits to be needed for managing the IUI.

There are many interesting directions left as future work. 
One possible direction is to consider CQI feedback. The acquired CQI can be exploited for sophisticated user scheduling, creating a link between the user scheduling and the used feedback bits.
Another direction is assuming a multi-antenna heterogeneous network with limited feedback. Since it was shown that the SIR coverage probability is related to the BS densities \cite{6287527} in  heterogeneous networks, the optimum number of feedback bits will likely relate to the BS densities. It would be also interesting to consider a cooperative multi-antenna cellular network with finite feedback. In a cooperative network, the required amount of limited feedback is a function of the BS cluster size \cite{NY:dynamic}. Another interesting direction is to consider different sources of CSIT inaccuracy. For example, the effect of feedback delay \cite{6177666, 6656848} or channel estimation error \cite{1193803} can be characterized with a similar performance metric as in this paper. 
In addition to that, 
generalizing to allow other numbers of users, e.g. $1 < K < N$, and more sophisticated beamforming strategies are also of interest.

\appendices
\section{Proof of Theorem 1} \label{proof:thm1}
We start the proof by rewriting the CCDF of the instantaneous SIR of the single-user case \eqref{sir_ccdf_su}
\begin{align} \label{theo_sir_ccdf_su_derivation1}
&\mathbb{P}\left[{\rm SIR_{SU}} > \gamma \right] \nonumber \\
&= \mathbb{P}\left[\left\|{\bf{h}}_1 \right\| ^2\left|\tilde {\bf{h}}_1^*  \hat {\bf{h}}_1 \right|^2> \gamma \left\| {\bf{d}}_1 \right\|^{\beta}   { \sum_{i=2}^{\infty}\left\| {\bf{d}}_i \right\|^{-\beta} \left| {\bf{h}}_{1,i}^*{\bf{v}}_i \right|^{2} }\right] \nonumber \\
& \mathop {=} \limits^{(a)} \mathbb{P}\left[\left\|{\bf{h}}_1 \right\| ^2> \gamma \frac{1}{\cos^2\theta_1}\left\| {\bf{d}}_1 \right\|^{\beta}   { \sum_{i=2}^{\infty}\left\| {\bf{d}}_i \right\|^{-\beta} \left| {\bf{h}}_{1,i}^*{\bf{v}}_i \right|^{2} }\right],
\end{align}
where (a) follows ${\rm sin^2}\theta_1 = 1 - \left|\tilde {\bf{h}}_1^* \hat{\bf{h}}_1 \right|^2$. Define 
$I = \left\| {\bf{d}}_1 \right\|^{\beta}   { \sum_{i=2}^{\infty}\left\| {\bf{d}}_i \right\|^{-\beta} \left| {\bf{h}}_{1,i}^*{\bf{v}}_i \right|^{2} }$ and $\mathcal{L}_X(s) = \mathbb{E}_{X}\left[e^{-sX} \right]$.
Then, we rewrite \eqref{theo_sir_ccdf_su_derivation1} as
\begin{align}
&\mathbb{P}\left[\left\|{\bf{h}}_1 \right\| ^2>  \frac{I \gamma}{\cos^2 \theta_1}\right]  \nonumber \\
& \mathop {=} \limits^{(a)} \mathbb{E}\left[ \mathbb{E}\left[\left.   \sum_{m=0}^{N-1} \frac{\gamma^m}{m!}  \frac{I^m }{\cos^{2m}\theta_1 } \exp\left( -\frac{I \gamma}{\cos^2 \theta_1} \right) \right| \cos^2 \theta_1 , I \right] \right] \nonumber \\
&\mathop = \limits^{(b)} \sum_{m=0}^{N-1} \frac{\gamma^m}{m!} (-1)^m \left. \frac{\partial^m \mathcal{L}_{I/\cos^2\theta_1}(s)}{\partial s^m} \right| _{s = {\gamma}}, 
\end{align}
where (a) follows that $\left\| {\bf{h}}_1 \right\|^2$ follows the Chi-squared distribution with $2N$ degrees of freedom and (b) follows the derivative property of the Laplace transform, which is
$\mathbb{E}\left[ X^m e^{-sX}\right] = (-1)^m \partial^m \mathcal{L}_X (s) / \partial s^m$. Now we obtain $\mathcal{L}_{I/\cos^2\theta_1}(s)$. 
\begin{align} \label{laplace_su_cos}
&\mathcal{L}_{I/\cos^2\theta_1}(s) \nonumber \\
&= \mathbb{E}\left[e^{-\frac{s}{\cos^2 \theta_1}{\left\| {\bf{d}}_1 \right\|^{\beta}   { \sum_{i=2}^{\infty}\left\| {\bf{d}}_i \right\|^{-\beta} \left| {\bf{h}}_{1,i}^*{\bf{v}}_i \right|^{2} }}{}} \right], {\rm }\; z = \frac{s}{\cos^2\theta_1}, \nonumber \\
& \mathop {=} \limits^{(a)} \mathbb{E}_{R, \cos^2\theta_1}\Bigg[ \mathbb{E}_{\Phi \backslash \mathcal{B}\left(0, R \right)}\Bigg[ \nonumber \\
& \;\;\;\;\;\;\;\;\;\;\;\;\;\; \left. \prod_{{\bf{d}}_i \in \Phi \backslash \mathcal{B}\left(0, R \right)}
\frac{1}{1+zR^{\beta}\left\|{\bf{d}}_i  \right\|^{-\beta}}
\right| \left\| {\bf{d}}_1 \right\| = R, \cos^2\theta_1 \Bigg] \Bigg]  \nonumber \\
&\mathop {=} \limits^{(b)} \mathbb{E}_{R, \cos^2\theta_1} \left[  \exp \left(-2\pi \lambda  \int_{R}^{\infty} \frac{z R^{\beta}r^{-\beta+1}}{1 + z  R^{\beta} r^{-\beta}} {\rm d} r\right) \right] \nonumber \\
&= \mathbb{E}_{R, \cos^2\theta_1} \left[ \exp \left(-2\pi \lambda R^2 {z} \frac{{}_2F_1\left(1, \frac{-2+\beta}{\beta}, 2- \frac{2}{\beta}, -z \right)}{\beta-2} \right)  \right] \nonumber \\
&\mathop {=} \limits^{(c)} \mathbb{E}_{\cos^2\theta_1} \left[ \frac{\beta - 2}{\beta - 2 + 2z \cdot {}_2F_1\left(1, \frac{-2+\beta}{\beta}, 2- \frac{2}{\beta}, -z \right)} \right] \nonumber \\
&\mathop {=} \limits^{} \mathbb{E}_{\cos^2\theta_1} \Bigg[ \nonumber \\
&\;\;\;\;\;\;\;\;\;\;\;\;\;\frac{\beta - 2}{\beta - 2 + 2\frac{s}{\cos^2\theta_1} \cdot {}_2F_1\left(1, \frac{-2+\beta}{\beta}, 2- \frac{2}{\beta}, -\frac{s}{\cos^2\theta_1} \right)} \Bigg],
\end{align} 
where (a) follows that $\left| {\bf{h}}_{1,i}^*{\bf{v}}_i \right|^{2} = H_i \sim \exp\left(1\right)$ due to the random beamforming effect, (b) follows the probability generating functional (PGFL) of PPP, and (c) takes the expectation over $R$ from its probability density function (PDF)
\begin{align} \label{first_touch}
f_{\left\| {\bf{d}}_1\right\|}(r) = 2\lambda \pi r e^{-\lambda \pi r^2}.
\end{align}
By leveraging the PDF of $\sin^2\theta_1$, which is obtained by differentiating the CDF of $\sin^2\theta_1$ \eqref{sin_cdf}, we finally get an integral form of $\mathcal{L}_{I/\cos^2\theta_1}(s)$ as follows.
\begin{align}
&\mathcal{L}_{I/\cos^2\theta_1}(s) \nonumber \\
&= \mathbb{E}_{\cos^2\theta_1} \left[ \frac{\beta - 2}{\beta - 2 + 2\frac{s}{\cos^2\theta_1} \cdot {}_2F_1\left(1, \frac{-2+\beta}{\beta}, 2- \frac{2}{\beta}, -\frac{s}{\cos^2\theta_1} \right)} \right] \nonumber \\
&=  \mathbb{E}_{\sin^2\theta_1} \Bigg[ \nonumber \\
&\;\;\;\;\;\;\;\;\;\;\;\; \frac{\beta - 2}{\beta - 2 + 2\frac{s}{1 - \sin^2\theta_1} \cdot  {}_2F_1\left(1, \frac{-2+\beta}{\beta}, 2- \frac{2}{\beta}, -\frac{s}{1-\sin^2\theta_1} \right)} \Bigg] \nonumber \\
&= \int_{0}^{2^{-\frac{B}{N-1}}} \frac{(\beta - 2) \cdot 2^B (N-1)x^{N-2} }{\beta - 2 + 2\frac{s}{1 - x} \cdot {}_2F_1\left(1, \frac{-2+\beta}{\beta}, 2- \frac{2}{\beta}, -\frac{s}{1-x} \right)} {\rm d}x. 
\end{align}
\hfill $\square$

\section{Proof of Corollary 1} \label{proof:coro1}
The proof relies on Lemma 2 in \cite{lee:spectral}, which provides the following lower bound.
\begin{align} \label{ny_lowerbound}
\log_2\left(1+ \frac{e^{\ln \mathbb{E}\left[X \right]}}{\mathbb{E}[Y]} \right) \le \log_2\left(1+\frac{X}{Y} \right),
\end{align}
where $X$ and $Y$ are non-negative random variables. By using this, we have
\begin{align}
&\mathbb{E}\left[\log_2\left(1+{\rm{SIR}}_{{\rm{MRT}}} \right) \right] \nonumber \\
& = \mathbb{E}\left[\log_2\left(1+\frac{  \left\|{\bf{h}}_1 \right\|^2 \cos^2 \theta_1 }{ \left\| {\bf{d}}_1 \right\|^{\beta} \sum_{i=2}^{\infty}\left\| {\bf{d}}_i \right\|^{-\beta} \left| {\bf{h}}_{1,i}^*{{\bf{v}}}_i \right|^{2} } \right) \right] \nonumber \\
&\ge \log_2\left(1+\frac{\exp\left(\mathbb{E}\left[\ln(  \left\|{\bf{h}}_1 \right\|^2 \cos^2 \theta_1) \right] \right)}{\mathbb{E}\left[\left\| {\bf{d}}_1 \right\|^{\beta} \sum_{i=2}^{\infty}\left\| {\bf{d}}_i \right\|^{-\beta} \left| {\bf{h}}_{1,i}^*{{\bf{v}}}_i \right|^{2} \right]} \right). \label{lower_sing_coro1}
\end{align}
Calculating the numerator inside the log function in \eqref{lower_sing_coro1}, we obtain
\begin{align}
\mathbb{E}\left[\ln(  \left\|{\bf{h}}_1 \right\|^2 \cos^2 \theta_1) \right] &= \underbrace{\mathbb{E}\left[\ln(  \left\|{\bf{h}}_1 \right\|^2 ) \right]}_{(a)} + \underbrace{\mathbb{E}\left[\ln \left(\cos^2 \theta_1 \right) \right] }_{(b)}.
\end{align}
Since $\left\| {\bf{h}}_1 \right\|^2$ follows the Chi-square distribution with degrees of freedom $2N$, $({\rm a}) = \psi(N)$ where $\psi(\cdot)$ is the digamma function defined as
\begin{align}
\psi(x) = \int_{0}^{\infty}\frac{e^{-t}}{t} - \frac{e^{-xt}}{1-e^{-t}} {\rm d} t.
\end{align}
Furthermore, from \eqref{sin_cdf}, we obtain 
\begin{align}
\mathbb{E}\left[\ln \left(\cos^2 \theta_1 \right) \right] = 2^{B} \left({\rm{B}}\left(2^{-\frac{B}{N-1}},N,0 \right)+ 2^{-{B}}\ln\left( 1-2^{-\frac{B}{N-1}}\right) \right),
\end{align}
where ${\rm B}\left(\cdot, \cdot, \cdot\right)$ is the incomplete Beta function defined as
\begin{align}
{\rm B}\left(z, a, b \right) = \int_{0}^{z}t^{a-1}(1-t)^{b-1} {\rm d} t.
\end{align}
Since $2^{-\frac{B}{N-1}}<1$, ${\rm{B}}\left(2^{-\frac{B}{N-1}},N,0 \right) \ge 0$, therefore ${\rm (b)} \ge \ln\left( 1-2^{-\frac{B}{N-1}}\right)$.
Calculating the denominator in \eqref{lower_sing_coro1}, we obtain
$\mathbb{E}\left[\left\| {\bf{d}}_1 \right\|^{\beta} \sum_{i=2}^{\infty}\left\| {\bf{d}}_i \right\|^{-\beta} \left| {\bf{h}}_{1,i}^*{{\bf{v}}}_i \right|^{2} \right] = \frac{2}{\beta-2}$. 
Plugging the results of the calculation into \eqref{lower_sing_coro1}, the following lower bound is established.
\begin{align}
&\mathbb{E}\left[\log_2\left(1+{\rm{SIR}}_{{\rm{MRT}}} \right) \right] \nonumber \\
&\ge \log_2\left(1+\frac{\exp\left(\psi(N) + \ln\left( 1-2^{-\frac{B}{N-1}}\right) \right)}{2/(\beta-2)} \right) \nonumber \\ 
&= \log_2\left(1+\left(1-2^{-\frac{B}{N-1}}\right)\frac{\exp\left(\psi(N)  \right)}{2/(\beta-2)} \right).
\end{align} 
\hfill $\square$

\section{Proof of Theorem 3} \label{proof:thm3}
The optimum feedback bits $B^{\star}_{\rm MRT}$ satisfies 
\begin{align} \label{thm3_argmax}
B^{\star}_{\rm MRT} = \mathop {\arg \max} \limits_{B \in \mathbb{N} \cup 0} \left( \mathbb{E}\left[\log_2 \left(1 + {\rm{SIR}}_{\rm MRT} \right) \right] - B/T_{\rm c} \right).
\end{align}
Solving \eqref{thm3_argmax} requires NP-hard complexity since it is in a class of integer programming. To resolve this, we first relax the feasible field of $B$ to the real number, i.e., $B \in \mathbb{R}$. Later, we can turn back to the original feasible field of $B$ by $\lceil B \rceil$.
After relaxation, the optimum feedback bits $B^{\star}_{\rm MRT}$ satisfies the KKT condition of \eqref{thm3_argmax}, which is given by
\begin{align} \label{theo_low_sin1}
\left. \frac{\partial \mathbb{E}\left[\log_2 \left(1 + {\rm{SIR}}_{\rm MRT} \right) \right] }{\partial B}\right |_{B=B^{\star}_{\rm MRT}}  = 1/T_{\rm c}.
\end{align}
It is worthwhile to mention that ${\partial \mathbb{E}\left[\log_2 \left(1 + {\rm{SIR}}_{\rm MRT} \right) \right] }/{\partial B}$ can be interpreted as the feedback efficiency, since it means how much the ergodic spectral efficiency increases when increasing small amount of  feedback. Calculating ${\partial \mathbb{E}\left[\log_2 \left(1 + {\rm{SIR}}_{\rm MRT} \right) \right] }/{\partial B}$, we have
\begin{align} \label{theo_low_sin2}
&\frac{\partial \mathbb{E}\left[\log_2 \left(1 + {\rm{SIR}}_{\rm MRT} \right) \right] }{\partial B} \nonumber \\
&= 2^{-\frac{B}{N-1}}\int_{0}^{\infty}\frac{\beta - 2}{\beta - 2 + 2z \cdot {}_2F_1\left(1, \frac{-2+\beta}{\beta}, 2- \frac{2}{\beta}, -z \right)} \cdot \nonumber \\
&\;\;\;\;\;\;\;\;\;\;\;\;\;\;\;\;\;\;\; \left( \frac{1}{1+z}\right) \left(\frac{1}{1 + z\left(1 - 2^{-\frac{B}{N-1}} \right)} \right)^N {\rm d} z.
\end{align}
After plugging \eqref{theo_low_sin2} into ${\partial \mathbb{E}\left[\log_2 \left(1 + {\rm{SIR}}_{\rm MRT} \right) \right] }/{\partial B}$, solving \eqref{theo_low_sin1} about $B$ provides the optimum feedback bits $B^{\star}_{\rm MRT}$. Finding a solution in \eqref{theo_low_sin1} is not straightforward, however, since \eqref{theo_low_sin2} does not have a closed form expression. For this reason, instead of directly solving \eqref{theo_low_sin1}, we rather consider a lower bound on ${\partial \mathbb{E}\left[\log_2 \left(1 + {\rm{SIR}}_{\rm MRT} \right) \right] }/{\partial B}$. Denoting 
$(\cdot)_{\rm LB}$ as a lower bound on the corresponding term inside the parenthesis, we consider the following problem in place of \eqref{theo_low_sin1}.
\begin{align} \label{theo_low_sin3}
\left( \frac{\partial \mathbb{E}\left[\log_2 \left(1 + {\rm{SIR}}_{\rm MRT} \right) \right] }{\partial B} \right)_{\rm LB}  = 1/T_{\rm c}.
\end{align}
One important point here is that since $\frac{{\partial \mathbb{E}\left[\log_2 \left(1 + {\rm{SIR}}_{\rm MRT} \right) \right] }}{{\partial B}}$ is a monotonically decreasing function of $B$, a solution of \eqref{theo_low_sin3} provides a lower bound on the optimum number of  feedback bits.

Now, we focus on obtaining $\left(\frac{ {\partial \mathbb{E}\left[\log_2 \left(1 + {\rm{SIR}}_{\rm MRT} \right) \right] }}{{\partial B}} \right)_{\rm LB}$. To do this, we
define the following random variables: $I = \left\| {\bf{d}}_1 \right\|^{\beta}\sum_{i=2}^{\infty}\left\| {\bf{d}}_i \right\|^{-\beta} \left| {\bf{h}}_{1,i}^*{\bf{v}}_i \right|^{2},  E \sim \rm{exp}\left(1\right),  G \sim \Gamma\left(N, 1 - 2^{-\frac{B}{N-1}} \right)$,
where $\Gamma\left(\cdot, \cdot\right)$ is the Gamma distribution.
With simple calculations, the expectations of the defined random variables are 
$\mathbb{E}\left[I \right] = \frac{2}{\beta - 2}, \mathbb{E}\left[E \right] = 1, \mathbb{E}\left[ G \right] = N\left(1 - 2^{-\frac{B}{N-1}} \right)$.
Now we rewrite \eqref{theo_low_sin2} with $I,E$, and $G$.
\begin{align}
& 2^{-\frac{B}{N-1}}\int_{0}^{\infty}\frac{\beta - 2}{\beta - 2 + 2z \cdot {}_2F_1\left(1, \frac{-2+\beta}{\beta}, 2- \frac{2}{\beta}, -z \right)} \cdot \nonumber \\
&\;\;\;\;\;\;\;\;\;\;\;\;\;\;\;\;\;\;\;\;\;\;\;\;\;\;\; \left( \frac{1}{1+z}\right) \left(\frac{1}{1 + z\left(1 - 2^{-\frac{B}{N-1}} \right)} \right)^N {\rm d} z \nonumber \\
&= 2^{-\frac{B}{N-1}}\int_{0}^{\infty}\mathbb{E} \left[e^{-zI} \right]\mathbb{E} \left[e^{-zE} \right]\mathbb{E} \left[e^{-zG} \right] {\rm d} z \nonumber \\
&\mathop {\ge} \limits^{(a)} 2^{-\frac{B}{N-1}}\int_{0}^{\infty} e^{-z\mathbb{E}\left[I\right]} e^{-z\mathbb{E}\left[E\right]} e^{-z\mathbb{E}\left[G\right]} {\rm d} z \nonumber \\
&\mathop = \limits^{(b)} 2^{-\frac{B}{N-1}}\int_{0}^{\infty} e^{-z\frac{2}{\beta - 2}} e^{-z} e^{-zN\left(1-2^{-\frac{B}{N-1}} \right)}  {\rm d} z \nonumber \\
& = \frac{2^{-\frac{B}{N-1}}}{\frac{2}{\beta - 2} + 1 + N\left(1 - 2^{-\frac{B}{N-1}} \right)},\label{theo_low_sin_low}
\end{align}
where (a) follows Jensen's inequality and (b) comes from the expectations of the defined random variables $I, E,$ and $G$.
Plugging \eqref{theo_low_sin_low} into $\left( \frac{{\partial \mathbb{E}\left[\log_2 \left(1 + {\rm{SIR}}_{\rm MRT} \right) \right] }}{{\partial B}} \right)_{\rm LB}$ in \eqref{theo_low_sin3}, 
we have
\begin{align}
\frac{2^{-\frac{B}{N-1}}}{\frac{2}{\beta - 2} + 1 + N\left(1 - 2^{-\frac{B}{N-1}} \right)} = 1/T_{\rm c},
\end{align}
which has a solution
\begin{align}
B^{\star}_{\rm L, MRT} =  \left( N-1\right) \log_2\left(\frac{\left(\beta - 2\right) N + (\beta -2)T_{\rm c}}{\left(\beta - 2 \right)N + \beta} \right). 
\end{align}
\hfill $\square$

\section{Proof of Theorem 4} \label{proof:thm4}
Before starting the proof, we denote $H_i$, $i \in \{1,...,K\}$ as a random variable that follows the exponential distribution with unit mean.
Then, by Lemma \ref{lem_beta}, the CCDF of the instantaneous SIR is rewritten as
\begin{align} \label{theo_sir_ccdf_mu_derivation}
& \mathbb{P}\left[{\rm SIR}_{\rm ZF}^k > \gamma \right] \nonumber \\
&= \mathbb{P}\left[ \frac{\left\| {\bf{h}}_k \right\|^2 \beta\left(1, N-1\right) }{\Bigg(\begin{array}{c} \left\| {\bf{h}}_k \right\|^2 {\rm sin}^2 \theta_k  \sum_{ k' =1,k' \ne k}^{K} \beta\left(1, N-2 \right)  \\ +   \left\| {\bf{d}}_1 \right\|^{\beta}  \sum_{i=2}^{\infty}\left\| {\bf{d}}_i \right\|^{-\beta} \left\| {\bf{h}}_{k,i}^*{\bf{V}}_i \right\|^{2}\end{array} \Bigg)}  > \gamma \right] \nonumber \\
&\mathop {=} \limits^{(a)} \mathbb{P}\left[ \frac{H_k}{\Bigg(\begin{array}{c}2^{-\frac{B}{N-1}}  \sum_{ k' =1,k' \ne k}^{K} H_{k'}  \\+   \left\| {\bf{d}}_1 \right\|^{\beta}  \sum_{i=2}^{\infty}\left\| {\bf{d}}_i \right\|^{-\beta} \left\| {\bf{h}}_{k,i}^*{\bf{V}}_i \right\|^{2}\end{array}\Bigg)} > \gamma \right] \nonumber \\
&\mathop {=} \limits^{(b)} \mathbb{E}\Bigg[\exp\Bigg(-\gamma \Bigg(2^{-\frac{B}{N-1}}\sum_{ k' =1,k' \ne k}^{K} H_{k'}  \nonumber \\
&\;\;\;\;\;\;\;\;\;\;\;\;\;\;\;\;\;\;\;\;\;\;\;\;\;\;\;\;\;\;\;\; + \left\| {\bf{d}}_1 \right\|^{\beta}  \sum_{i=2}^{\infty}\left\| {\bf{d}}_i \right\|^{-\beta} \left\| {\bf{h}}_{k,i}^*{\bf{V}}_i \right\|^{2} \Bigg) \Bigg)  \Bigg] \nonumber \\
&\mathop {=} \limits^{(c)} \mathbb{E}\left[\exp\left(-\gamma 2^{-\frac{B}{N-1}}\sum_{ k' =1,k' \ne k}^{K} H_{k'} \right) \right]\cdot  \nonumber \\
&\;\;\;\;\;\;\;\;\;\;\;\;\;\;\;\;\; \mathbb{E}\left[\exp\left(-\gamma \left\| {\bf{d}}_1 \right\|^{\beta}  \sum_{i=2}^{\infty}\left\| {\bf{d}}_i \right\|^{-\beta} \left\| {\bf{h}}_{k,i}^*{\bf{V}}_i \right\|^{2}\right) \right],
\end{align}
where (a) follows that $\left\| {\bf{h}}_k \right\|^2 {\rm sin}^2 \theta_k \beta \left(1, N-2\right)$ has a distribution of $\Gamma(1, \delta)$ with $\delta = 2^{-\frac{B}{N-1}}$ by Lemma 1 in \cite{4299617} and Lemma \ref{lem_beta}, (b) follows that $H_k$ is an exponential random variable with unit mean, and (c) comes from the independence between the IUI and ICI terms. Note that \eqref{theo_sir_ccdf_mu_derivation} is a multiplication of the Laplace transforms of IUI and ICI. 
The Laplace transform of the IUI is
\begin{align} \label{laplace_iui_sir_ccdf_final}
\mathcal{L}_{I_{\rm U}}(s)
&\mathop {=} \limits^{(a)} \prod_{k'=1 k'\neq k}^{K} \mathbb{E}\left[\exp\left(-s 2^{-\frac{B}{N-1}}H_{k'} \right) \right] \nonumber \\
&\mathop {=} \limits^{(b)} \left( \frac{1}{1+s 2^{-\frac{B}{N-1}}} \right)^{N-1},
\end{align}
where (a) comes from our assumption that each of IUI term is independent and (b) follows the Laplace transform of an exponential random variable with unit mean and the assumption that $N=K$.
The Laplace transform of the ICI is 
\begin{align} 
&\mathcal{L}_{I_{\rm C}}(s) = \mathbb{E}\left[\exp\left(-s \left\| {\bf{d}}_1 \right\|^{\beta}  \sum_{i=2}^{\infty}\left\| {\bf{d}}_i \right\|^{-\beta} \left\| {\bf{h}}_{k,i}^*{\bf{V}}_i \right\|^{2}\right) \right] \nonumber \\
&\mathop {=} \limits^{(a)} \mathbb{E}_{R}\Bigg[\left. \right. \mathbb{E}_{\Phi \backslash \mathcal{B}(0,R)}\Bigg[ \nonumber \\
&\;\;\;\;\;\;\;\;\;\;\;\;\;\;\; \left. \prod_{{\bf{d}}_i \in \Phi \backslash \mathcal{B}(0,R)} \left( \frac{1}{1+sR^{\beta}\left\| {\bf{d}}_i\right\|^{-\beta}}\right)^K  \right| \left\| {\bf{d}}_1\right\| = R \Bigg]  \Bigg] \nonumber \\
&\mathop {=} \limits^{(b)} \mathbb{E}_{R}\left[ \exp\left(-2\pi \lambda  \int_{R}^{\infty} \left(1- \left(\frac{1}{1+sR^{\beta}r^{-\beta}} \right)^K \right)r{\rm d} r \right)  \right] \nonumber 
\end{align}
\begin{align} \label{laplace_ici_sir_ccdf_final}
&\mathop {=} \limits^{(c)} \mathbb{E}_{R}\left[ \exp\left(-2\pi \lambda R^2 \int_{1}^{\infty} \left(1-\left(\frac{1}{1+st^{-\beta}} \right)^K \right)t{\rm d} t \right) \right] \nonumber \\
&\mathop {=} \limits^{(d)} \frac{1}{1+\left(2\int_{1}^{\infty} \left(1-\left(\frac{1}{1+st^{-\beta}} \right)^K \right)t{\rm d} t \right)} \nonumber \\
&\mathop {=} \limits^{(e)} \frac{1}{{}_2F_1\left(N, -\frac{2}{\beta}, 1-\frac{2}{\beta}, -s \right)},
\end{align}
where (a) follows $\left\| {\bf{h}}_{k,i}^*{\bf{V}}_i \right\|^{2}$ is a Chi-squared random variable with $2K$ degrees of freedom, i.e., $\chi_{2K}^{2}$, (b) follows the PGFL of PPP, (c) comes from the variable change $t=R^{-1}r$, (d) follows the first-touch distribution of PPP \eqref{first_touch}, and (e) comes from the variable change $st^{-\beta} = u$ and the assumption that $K=N$.
Plugging \eqref{laplace_iui_sir_ccdf_final} and \eqref{laplace_ici_sir_ccdf_final} into \eqref{theo_sir_ccdf_mu_derivation}, we obtain
\begin{align}
&\mathbb{P}\left[{\rm SIR}_{\rm ZF}^k > \gamma \right] \nonumber \\
&= \left( \frac{1}{1+\gamma  2^{-\frac{B}{N-1}}} \right)^{N-1} \frac{1}{{}_2F_1\left(K, -\frac{2}{\beta}, 1-\frac{2}{\beta}, -\gamma \right)}.
\end{align}
\hfill $\square$

\section{Proof of Theorem 6} \label{proof:thm6}
Following the same logic in the proof of Theorem \ref{theo_low_sing}, instead of finding an exact expression, we rather obtain a lower bound on $\frac{\partial \mathbb{E}\left[\log_2\left(1 + {\rm SIR}_{\rm ZF}^{k} \right) \right] }{ \partial B}$ assuming $T_{\rm c} \rightarrow \infty$ thereby $B^{\star}_{\rm ZF} \rightarrow \infty$. 
\begin{align} \label{them_low_approx_derivative}
&\frac{\partial \mathbb{E}\left[\log_2 \left(1 + {\rm{SIR}}^{k}_{\rm ZF} \right) \right] }{\partial B} \nonumber \\
&= 2^{-\frac{B}{N-1}} \int_{0}^{\infty}  \left( \frac{z}{1+z}\right)\left(\frac{1}{1+z 2^{-\frac{B}{N-1}}} \right)^N \cdot \nonumber \\
&\;\;\;\;\;\;\;\;\;\;\;\;\;\;\;\;\;\;\;\;\;\;\;\;\;\;\;\;\;\;\;\;\;\;\;\;\; \left(\frac{1}{{}_2F_1\left(N, -\frac{2}{\beta}, 1-\frac{2}{\beta}, -z \right)} \right)  {\rm d} z \nonumber \\
& \mathop {\ge}^{(a)} 2^{-\frac{B}{N-1}} \int_{0}^{\infty}  \left( \frac{z}{1+z}\right)\exp\left(-z  N 2^{-\frac{B}{N-1}} \right) \cdot \nonumber \\
&\;\;\;\;\;\;\;\;\;\;\;\;\;\;\;\;\;\;\;\;\;\;\;\;\;\;\;\;\;\;\;\;\;\;\;\;\;\; \left(\frac{1}{1+z^{\frac{2}{\beta}}\frac{2}{\beta}N^{\frac{2}{\beta}}\left(-\Gamma(-\frac{2}{\beta}) \right)}  \right)  {\rm d} z,
\end{align}
where (a) follows Lemma \ref{lem_lower_MGF_ICI} and the inequality of the exponential function $1/(1+x) \ge e^{-x}$ for $x>0$. 
We calculate \eqref{them_low_approx_derivative} by separating the integration range into $\left[0,C\right)$ and $\left[C,\infty\right)$. For $C<\infty$,
\begin{align} 
&2^{-\frac{B}{N-1}} \int_{0}^{C}  \left( \frac{z}{1+z}\right)\exp\left(-z  N 2^{-\frac{B}{N-1}} \right) \cdot \nonumber \\
&\;\;\;\;\;\;\;\;\;\;\;\;\;\;\;\;\;\;\;\;\;\;\;\;\;\;\;\;\;\;\;\;\;\;\;\;\;\;  \left(\frac{1}{1+z^{\frac{2}{\beta}}\frac{2}{\beta}N^{\frac{2}{\beta}}\left(-\Gamma(-\frac{2}{\beta}) \right)}  \right) {\rm d} z \label{thm_approx_low_mu1} \\
& \le 2^{-\frac{B}{N-1}} \int_{0}^{C}  \exp\left(-z  N 2^{-\frac{B}{N-1}} \right) \left(\frac{1}{z^{\frac{2}{\beta}}\frac{2}{\beta}N^{\frac{2}{\beta}}\left(-\Gamma(-\frac{2}{\beta}) \right)}  \right) {\rm d} z \nonumber 
\end{align}
\begin{align}
&= 2^{-\frac{B}{N-1}} \frac{ \left( \Gamma\left(\frac{-2+\beta}{\beta}\right) - \Gamma\left(\frac{-2+\beta}{\beta}, 2^{-\frac{B}{N-1}}CN \right)  \right) } {\frac{2}{\beta} N^{\frac{2}{\beta}} \left(-\Gamma\left(-\frac{2}{\beta} \right) \right)\left(2^{-\frac{B}{N-1}}N \right)^{1-\frac{2}{\beta}} } . \label{thm_approx_low_mu2}
\end{align}
Since $B^{\star}_{\rm ZF} \rightarrow \infty$, $2^{-\frac{B}{N-1}} \rightarrow 0$ for our interest $B$. Then we have
\begin{align}
2^{-\frac{B}{N-1}} \frac{ \left( \Gamma\left(\frac{-2+\beta}{\beta}\right) - \Gamma\left(\frac{-2+\beta}{\beta}, 2^{-\frac{B}{N-1}}CN \right)  \right) } {\frac{2}{\beta} N^{\frac{2}{\beta}} \left(-\Gamma\left(-\frac{2}{\beta} \right) \right)\left(2^{-\frac{B}{N-1}}N \right)^{1-\frac{2}{\beta}}} \rightarrow 0.
\end{align}
Since $\eqref{thm_approx_low_mu1}>0$, $\eqref{thm_approx_low_mu1} \rightarrow 0$ as $2^{-\frac{B}{N-1}}\rightarrow 0$. For this reason, as $2^{-\frac{B}{N-1}} \rightarrow 0$, we have
\begin{align} \label{them_low_part_a}
& 2^{-\frac{B}{N-1}} \int_{0}^{C}  \left( \frac{z}{1+z}\right)\exp\left(-z  N 2^{-\frac{B}{N-1}} \right) \cdot\nonumber \\
&\;\;\;\;\;\;\;\;\;\;\;\;\;\;\;\;\;\;\;\;\;\;\;\;\;\;\;\;\;\;\;\;\;\;\;\;\;\; \left(\frac{1}{1+z^{\frac{2}{\beta}}\frac{2}{\beta}N^{\frac{2}{\beta}}\left(-\Gamma(-\frac{2}{\beta}) \right)}  \right) {\rm d} z \nonumber \\
& \approx 2^{-\frac{B}{N-1}} \int_{0}^{C}  \exp\left(-z  N 2^{-\frac{B}{N-1}} \right) \cdot \nonumber \\
& \;\;\;\;\;\;\;\;\;\;\;\;\;\;\;\;\;\;\;\;\;\;\;\;\;\;\;\;\;\;\;\;\;\;\;\;\;\;\;\;\;\;\;\; \left(\frac{1}{z^{\frac{2}{\beta}}\frac{2}{\beta}N^{\frac{2}{\beta}}\left(-\Gamma(-\frac{2}{\beta}) \right)}  \right) {\rm d} z.
\end{align}
Now we compute the integration over $\left[C,\infty\right)$ of \eqref{them_low_approx_derivative}. For large enough $C$ and $z>C$, we have $\frac{z}{1+z} \approx 1$ and $\frac{1}{1+z^{\frac{2}{\beta}}\frac{2}{\beta}N^{\frac{2}{\beta}}\left(-\Gamma(-\frac{2}{\beta}) \right)}   \approx \frac{1}{z^{\frac{2}{\beta}}\frac{2}{\beta}N^{\frac{2}{\beta}}\left(-\Gamma(-\frac{2}{\beta}) \right)}$.
By leveraging these approximations, we have
\begin{align} \label{them_low_part_b}
&2^{-\frac{B}{N-1}} \int_{C}^{\infty}  \left( \frac{z}{1+z}\right)\exp\left(-z  N 2^{-\frac{B}{N-1}} \right) \cdot \nonumber \\
&\;\;\;\;\;\;\;\;\;\;\;\;\;\;\;\;\;\;\;\;\;\;\;\;\;\;\;\;\;\;\;\;\;\;\;\;\;\;\;\; \left(\frac{1}{1+z^{\frac{2}{\beta}}\frac{2}{\beta}N^{\frac{2}{\beta}}\left(-\Gamma(-\frac{2}{\beta}) \right)}  \right)  {\rm d} z \nonumber \\
&\approx 2^{-\frac{B}{N-1}} \int_{C}^{\infty}  \exp\left(-z  N 2^{-\frac{B}{N-1}} \right) \cdot \nonumber \\
&\;\;\;\;\;\;\;\;\;\;\;\;\;\;\;\;\;\;\;\;\;\;\;\;\;\;\;\;\;\;\;\;\;\;\;\;\;\;\;\;\;\;\;\;\;\;\; \left(\frac{1}{z^{\frac{2}{\beta}}\frac{2}{\beta}N^{\frac{2}{\beta}}\left(-\Gamma(-\frac{2}{\beta}) \right)}  \right)  {\rm d} z.
\end{align}
Combining \eqref{them_low_part_a} and \eqref{them_low_part_b}, \eqref{them_low_approx_derivative} is given as
\begin{align} \label{them_low_approx_final}
&\frac{\partial \mathbb{E}\left[\log_2 \left(1 + {\rm{SIR}}^{k}_{\rm ZF} \right) \right] }{\partial B} \nonumber \\
& \ge 
2^{-\frac{B}{N-1}} \int_{0}^{\infty}  e^{\left(-z  N 2^{-\frac{B}{N-1}} \right)} \left(\frac{1}{z^{\frac{2}{\beta}}\frac{2}{\beta}N^{\frac{2}{\beta}}\left(-\Gamma(-\frac{2}{\beta}) \right)}  \right) {\rm d} z \nonumber \\
&= \frac{\Gamma\left(1-\frac{2}{\beta} \right)}{\frac{2}{\beta}N\left(-\Gamma(-\frac{2}{\beta}) \right)
}  2^{-(\frac{2}{\beta})\frac{B}{N-1}}.
\end{align}
Now we consider the equation of $\eqref{them_low_approx_final} = 1/T_{\rm c}$,
and it has a solution at 
\begin{align}
\tilde B^{\star}_{\rm L,ZF} = \left(N-1 \right)\log_2\left(\frac{\beta \Gamma\left(1-\frac{2}{\beta}\right) T_{\rm c}}{ 2 N \left(-\Gamma(-\frac{2}{\beta}) \right)}\right)^{\frac{\beta}{2}}.
\end{align}
\hfill $\square$

\bibliographystyle{IEEEtran}
\bibliography{ref}

\end{document}